\documentclass[12pt]{amsart}

\usepackage{amsmath,amstext,amssymb,amsopn,amsthm}
\usepackage{url,verbatim,hyperref}
\usepackage{mathtools}
\usepackage{enumerate}

\usepackage{color,graphicx}

\usepackage[margin=30mm, marginparwidth=20mm]{geometry}
\usepackage{eucal,mathrsfs,dsfont}
\usepackage{caption}
\usepackage{subcaption}

\usepackage{fancyhdr}
\allowdisplaybreaks

\newtheorem{theorem}{Theorem}[section]
\newtheorem{corollary}[theorem]{Corollary}
\newtheorem{lemma}[theorem]{Lemma}
\newtheorem{proposition}[theorem]{Proposition}

\theoremstyle{definition}

\newtheorem{definition}[theorem]{Definition}

\newtheorem{remark}[theorem]{Remark}
\newtheorem{example}[theorem]{Example}

\numberwithin{equation}{section}

\newcommand{\eps}{\varepsilon}

\newcommand{\calA}{\mathcal{A}}

\newcommand{\calQ}{\mathcal{Q}}

\newcommand{\calT}{\mathcal{T}}

\newcommand{\calB}{\mathcal{B}}

\newcommand{\calR}{\mathcal{R}}

\newcommand{\E}{\operatorname{\mathds{E}}} 
\renewcommand{\P}{\operatorname{\mathds{P}}} 
\newcommand{\R}{\mathds{R}}

\newcommand{\prt}{\partial}

\newcommand{\wh}{\widehat}
\newcommand{\wt}{\widetilde}

\DeclareMathOperator{\sign}{sgn}

\def\bone{{\bf 1}}
\def\bn{{\bf n}}

\newcommand{\om}{\omega}
\newcommand{\Om}{\Omega}
\newcommand{\Ot}{\widetilde{\Omega}}
\newcommand{\al}{\alpha}
\newcommand{\si}{\sigma}
\newcommand{\de}{\delta}
\newcommand{\la}{\lambda}
\newcommand{\Si}{\Sigma}
\newcommand{\pa}{\partial}
\newcommand{\na}{\nabla}
\newcommand{\pb}{\overline{p}}
\newcommand{\vb}{\overline{v}}
 
\newcommand{\sbar}{\overline{\sigma}} 
 
\newcommand{\Phibar}{\overline{\Phi}}  
\newcommand{\cS}{\mathcal{S}}
\newcommand{\cM}{\mathcal{M}}
\newcommand{\cH}{\mathcal{H}} 
\newcommand{\cD}{\mathcal{D}}  

\newcommand{\ph}{\varphi}
\usepackage{bm}
\usepackage{marginnote}

\def\crn#1#2{{\vcenter{\vbox{
        \hbox{\kern#2pt \vrule width.#2pt height#1pt
           }
          \hrule height.#2pt}}}}
\def\intprod{\mathchoice\crn54\crn54\crn{3.75}3\crn{2.5}2}
\def\into{\mathbin{\intprod}}

\title{Fermi acceleration in rotating drums}
\author{Krzysztof Burdzy, Mauricio Duarte, Carl-Erik Gauthier, C. Robin
  Graham and Jaime San Martin}

\address{KB, CEG and CRG: Department of Mathematics, Box 354350, University of
  Washington, Seattle, WA 98195} 
\email{burdzy@uw.edu}
\email{carlerik.gauthier@gmail.com}
\email{robin@math.washington.edu}

\address{MD: Departamento de Matematicas, Facultad de Ciencias Exactas,
  Universidad Andres Bello, Santiago, Chile} 
   \email{mauricio.duarte@unab.cl}

\address{JSM: Universidad de Chile, Facultad de Ingenieria,
Depto. de Mathem\'aticas, Beaucheff 850, Santiago, Chile}
\email{jsanmart@dim.uchile.cl}

\thanks{KB's research was supported in part by Simons Foundation Grant 506732.  }



\begin{document}

\begin{abstract}
Consider hard balls in a bounded rotating drum.
If there is no gravitation then there is no Fermi acceleration, i.e., the
energy of the balls  remains bounded forever. If there is gravitation,
Fermi acceleration may arise. A number of explicit formulas for the system
without gravitation are given.  Some of these are based on an explicit
realization, which we derive, of the well-known microcanonical ensemble
measure. 
\end{abstract}

\maketitle

\section{Introduction}\label{intro}

We will discuss systems of hard balls or pointlike particles in rotating
drums. The main problem that we are going to address is whether there is
Fermi acceleration in a rotating drum, i.e., whether the energy of the
balls can go to infinity. The short answer is no, if there is no
gravitation, and yes, if there is gravitation. 

When there is no gravitation then the system has a conserved quantity and
its analysis 
follows established methods from classical mechanics. In this case, our
contributions consist of derivation of an explicit realization of the
microcanonical ensemble measure and explicit formulas for some quantities
characterizing the system.  

The case of a rotating drum with gravitation is much harder to analyze
explicitly, yet we provide some detailed calculations that support the
claim of Fermi acceleration in that case. 

In the rest of the introduction, we will discuss our models in more detail
and outline the main results. We will also provide a review of  related
literature. 

\subsection{Rotating drum}\label{RD}
Gas centrifuges used for separation of uranium isotopes have cylindrical
rotors. Before we discuss a mathematical model inspired by cylindrical
rotors, we first consider a  model inspired by a centrifuge with an
impeller (see Fig. \ref{fig1}). 

It comes at no additional cost to analyze a  general bounded
$d$-dimensional domain $D$, with $d\geq 2$, (see Fig. \ref{fig2}),
generalizing the centrifuge with impeller (Fig. \ref{fig1}).  In these
illustrations, $D$ consists of the part of the centrifuge in the exterior
of the impeller blades.  We assume that there exists a 2-dimensional  
subspace of $\R^d$ (the ``horizontal'' subspace) in which $D$ is rotating  
with constant angular velocity about $(0,0)$, while other  
coordinates remain constant. Gas molecules are represented by hard
balls. The collisions between the balls  are assumed to be totally elastic,
i.e., we assume conservation of total energy and momentum. The balls move
with constant velocities between collisions. A ball reflects from the
boundary of $D$ according to the classical specular reflection in the
moving frame of reference which makes the reflecting wall static at the
moment of collision.  
At a collision of a ball with the wall at a point where the horizontal
component of the normal to the wall is not radial, evidently  
the energy of the ball increases when it collides with   
the wall moving towards the ball, and the energy of the ball decreases when it
collides with the wall moving away from the ball.  
Since the energy of the balls goes up and down, in principle it is 
possible that it will grow to infinity (over the infinite time horizon),
i.e., the Fermi acceleration might occur; the term refers to a model
proposed by Fermi in \cite{Fermi49}. 
A simple mathematical version of the model was presented by Ulam in \cite{Ulam61}. 
\begin{figure}
\centering
\begin{minipage}{.45\textwidth}
  \centering
  \includegraphics[width=1\linewidth]{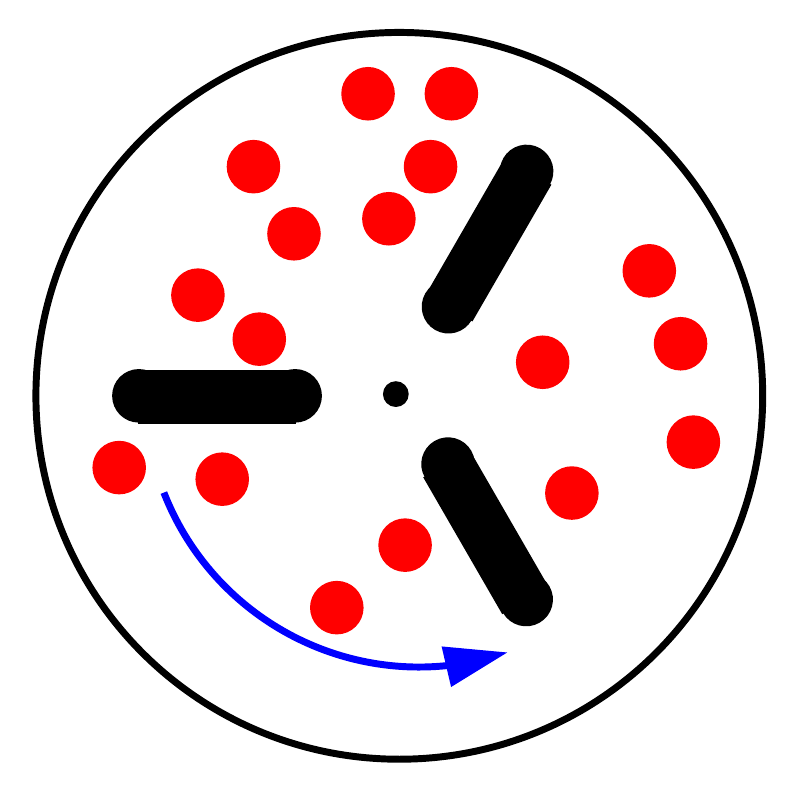}
  \captionof{figure}{Horizontal cross-section of vertical centrifuge with a three-blade impeller.}
  \label{fig1}
\end{minipage}%
\begin{minipage}{.45\textwidth}
  \centering
  \includegraphics[width=1\linewidth]{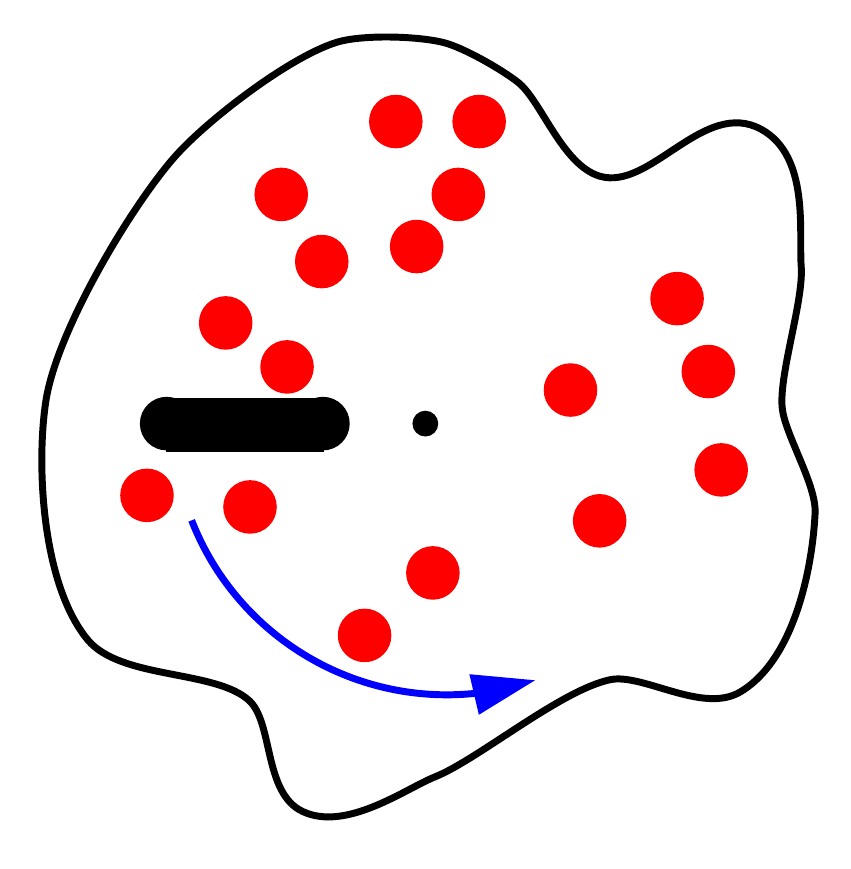}
  \captionof{figure}{Horizontal cross-section of a rotating drum. The
    rotating walls include the outside boundary of the domain and an
    internal impeller.} 
  \label{fig2}
\end{minipage}
\end{figure}
It turns out that the  dynamical system defined above has a conserved
quantity and, as a consequence, there is no Fermi acceleration in this 
case. We will now present a more formal definition of the system and the
conserved quantity.  

Let $H$ be the two-dimensional subspace of $\R^d$ spanned by the first two
basis vectors.  Consider a bounded $d$-dimensional  domain with smooth
boundary rotating with the constant 
angular velocity $\omega>0$ in  $H$ about the origin. Let $P_H$ denote the
projection on $H$ and let $\calR(\theta)$ denote the rotation by angle
$\theta$ in $H$.  Set $L=\calR(\pi/2)\circ P_H$.   
Suppose that the rotating drum $D$ holds $n$ balls, for some finite $n\ge 1$.
Let $F$ be the moving frame of reference rotating with the centrifuge. In 
other words, the drum $D$ is static in $F$. For the $k$-th molecule, let
$m_k$ be its mass, $v_k(t)\in \R^d$  its velocity at time $t$, 
$v^F_k(t)\in \R^d$  its velocity in $F$ at time $t$,  $x_k(t)\in \R^d$  its
center  at 
time $t$,  $x_k^F(t)\in \R^d$  its center in $F$ at time $t$,  
$x^H_k(t)\in H$ the projection of its center on $H$ at time $t$, and
$x^{F,H}_k(t)\in H$ the projection of its center on $H$ in $F$ at time $t$.   
Note that
$x^H_k(t) = P_H(x_k(t))$, $x^{F,H}_k(t) = P_H(x^F(t))$,
$x_k^F (t)  = \calR(-\omega t) (x_k(t) ) $, 
$\|x^{F,H}_k(t)\|=\|x^{H}_k(t)\|$.  The velocities are related by 
\begin{equation}\label{velocities}
\calR(\omega t)(v_k^F(t)) =v_k(t)  - \omega L(x_k(t)) 
\end{equation}
(see Section~\ref{fe16.1}).  If $d=3$, the angular velocity is typically  
represented as the vector 
$\bm{\omega}:=\omega\bm{e_3}$ and then  
$\omega L(x_k(t))=\bm{\omega}\times x_k(t)$.  Let 
\begin{align}\label{Fenergies}
E^{F,K}(t)&=\sum_{k=1}^n  m_k \|v_k^F(t)\|^2/2,\notag \\
E^{F,P}(t)&=-\sum_{k=1}^n  m_k \omega^2 \|x^H_k(t)\|^2/2,\\  
E^F(t) & = E^{F,K}(t) + E^{F,P}(t).\notag
\end{align}
We will show in Propositions~\ref{rotatingprop} and \ref{f17.1} that the
quantity $E^F(t) $ is  conserved, not only between collisions but also 
at collisions of balls or of a ball and a wall.  Thus   
$E^F :=E^F(t)=E^F(0)  $ for all $t\geq 0$.  This can be interpreted as the 
conservation of energy in the rotating 
(non-inertial) frame of reference $F$. The sum of the ``kinetic energy''
$E^{F,K}(t) $ and ``potential energy'' $E^{F,P}(t)$, i.e., 
$E^F = E^{F,K}(t) + E^{F,P}(t)$, remains constant. The potential energy
$E^{F,P}(t)$ is associated with the centrifugal force. There is no
contribution from the Coriolis force.  The usual kinetic energy in the
inertial frame is of course preserved between collisions, since the balls are 
free, and at collisions of two balls, since these collisions are assumed
elastic.  But, as noted above, it is not preserved at collisions of a 
ball and a wall. 

We will now sketch an argument that there is no Fermi acceleration, based
on the conservation of $E^F$. 
Since the drum is represented by a bounded domain $D$ and the angular
velocity $\omega$ is fixed, all quantities in the formula 
$-\sum_{k=1}^n m_k \omega^2 \|x^H_k(t)\|^2/2$ stay bounded and, 
therefore, the potential
energy $E^{F,P}(t)$ stays bounded. This and the fact that the energy $E^F$
is constant imply that the kinetic energy $E^{F,K}(t)$ must remain
bounded. Again since the domain $D$ is bounded and $\omega$ is fixed, the
difference  
$\|\calR(\omega t)(v_k^F(t)) - v_k(t)\|=\|\omega L(x_k(t))\|$ stays 
bounded. Thus the difference between $E^{F,K}(t)$ and the total energy in
the inertial frame of reference $E^K(t):=\sum_{k=1}^n m_k \|v_k(t)\|^2/2$
remains bounded and, therefore, $E^K(t)$ remains bounded. Hence, there is
no Fermi acceleration. 

Next we describe an invariant measure for this model, a special case of the
microcanonical ensemble (see \cite[Sect.~1.2]{Ruelle}).  
There are typically many invariant measures for the system defined above
but the measure we will present is unique for some probabilistic systems 
defined later.  

For $E^F\in \R$, we will give a formula for an invariant measure on the
level set  
\[
\cS_{E^F}=\Big\{\left(x_1^F ,x_2^F ,\dots,x_n^F ,v_1 ^F,v_2^F
,\dots,v_n^F\right):
\sum_{k=1}^n m_k \|v_k^F\|^2/2 -
\sum_{k=1}^n m_k\omega^2\|x^{F,H}_k\|^2/2=E^F\Big\}  
\]
of the conserved function 
in the space of
\begin{align*}
\left(x_1^F ,x_2^F ,\dots,x_n^F ,v_1 ^F,v_2^F ,\dots,v_n^F \right),
\end{align*}
not
$\left(x_1 ,x_2 ,\dots,x_n ,v_1 ,v_2 ,\dots,v_n \right)$.
Since the balls have positive radii, the vector of centers 
$\left(x_1^F ,x_2^F,\dots,x_n^F \right)$ is required 
to lie in the closed subset $\cD$ of $D^n$ defined by the conditions that
the interiors of the balls are contained in $D$ and do not overlap. 
Set 
$\cS_{E^F,\cD}=\left\{\left(x_1^F ,x_2^F ,\dots,x_n^F,
v_1^F,v_2^F,\dots,v_n^F \right)\in \cS_{E^F}:
\left(x_1^F ,x_2^F ,\dots,x_n^F\right)\in \cD\right\}$  
and denote by $\cS_{E^F}^0$, $\cS_{E^F,\cD}^0$ the subsets of
$\cS_{E^F}$, $\cS_{E^F,\cD}$, resp., for which    
$\left(v_1^F,v_2^F,\dots,v_n^F \right)\neq 0$; equivalently for which  
$E^F + \sum_{k=1}^n  m_k \omega^2 \|x^{F,H}_k \|^2/2> 0$.   
Consider the bijection 
$\Phibar:\cS^0_{E^F,\cD}\rightarrow \cD\times S^{nd-1}$ given by 
\[
\Phibar\left(x_1^F ,x_2^F ,\dots,x_n^F ,v_1^F ,v_2^F ,\dots,v_n^F \right) 
=\left(x_1^F ,x_2^F ,\dots,x_n^F ,\vb_1^F ,\vb_2^F
,\dots,\vb_n^F \right), 
\]
where $S^{nd-1}$ denotes the Euclidean unit sphere in $\R^{nd}$ and 
\begin{equation}\label{vbar}
\vb_k^F 
=
\frac{\sqrt{m_k}v_k^F}
{\left(\sum_{k=1}^nm_k\|v_k^F\|^2\right)^{1/2}}.
\end{equation}
On $\cS^0_{E^F,\cD}$, the invariant measure can be written as   
\begin{equation}\label{f13.10}
f\left(x_1^F ,x_2^F ,\dots,x_n^F \right)
\Phibar^*\left(d\sigma_1\left(\vb_1^F , 
\vb_2^F,\dots,\vb_n^F
\right)\,d\lambda\left( x_1^F ,x_2^F ,\dots,x_n^F \right)\right).  
\end{equation}
Here $d\sigma_1$ denotes the usual measure on the unit sphere $S^{nd-1}$,
$d\lambda$ denotes Lebesgue measure on $\R^{nd}$, 
$\Phibar^*\left(d\sigma_1\left(\vb_1^F
,\vb_2^F,\dots,\vb_n^F
\right)\,d\lambda
\left(x_1^F ,x_2^F ,\dots,x_n^F \right)\right)$ 
denotes the  pullback of the product measure, and the density function $f$
is given by 
\begin{equation}\label{f13.11}
f\left(x_1^F ,x_2^F ,\dots,x_n^F \right)
=\left(E^F + \sum_{k=1}^n  m_k \omega^2 \|x^{F,H}_k
\|^2/2\right)^{nd/2-1}.   
\end{equation}
In particular, given $E^{F,P}$ (equivalently, given $E^{F,K}$, since $E^F$
is fixed), 
the kinetic energy  of the system in the frame $F$ is equidistributed on
the sphere, i.e., the vector  
\begin{align}\label{f13.12}
\left(\sqrt{\frac{m_1}{2}}v_1^F ,\sqrt{\frac{m_2}{2}}v_2^F ,\dots,
\sqrt{\frac{m_n}{2}}v_n^F 
\right)
\end{align}
is uniformly distributed on the $(nd-1)$-dimensional sphere centered at  
0, with radius  $ \left(E^{F,K} \right)^{1/2}$. 

\subsection{Rotating drum with Lambertian reflections}
Our next model is concerned with a rotating drum with a rough surface. On
the mathematical side, this means that the  walls  
$\prt D$ are smooth but reflections are not necessarily specular (i.e., the
angle of reflection is not necessarily equal to the angle of incidence). We
will consider a special class of random reflections. There are two sources
of inspiration for this model. First, as we have already mentioned,  
gas centrifuges used for separation of uranium isotopes have cylindrical
rotors. The surface of the rotor has to be rough or otherwise the molecules
of gas would not have a tendency to rotate. The second reason for
introducing random  reflections is that the microcanonical ensemble  is not
the only invariant measure for the dynamical system described in Section
\ref{RD}. Fig. \ref{fig3} shows a trivial invariant measure represented by
a closed orbit of a single particle. The invariant measure is unique if the
reflections from the walls are random; more precisely, it is unique for
some random reflection laws. 

\begin{figure}
  \includegraphics[width=.4\linewidth]{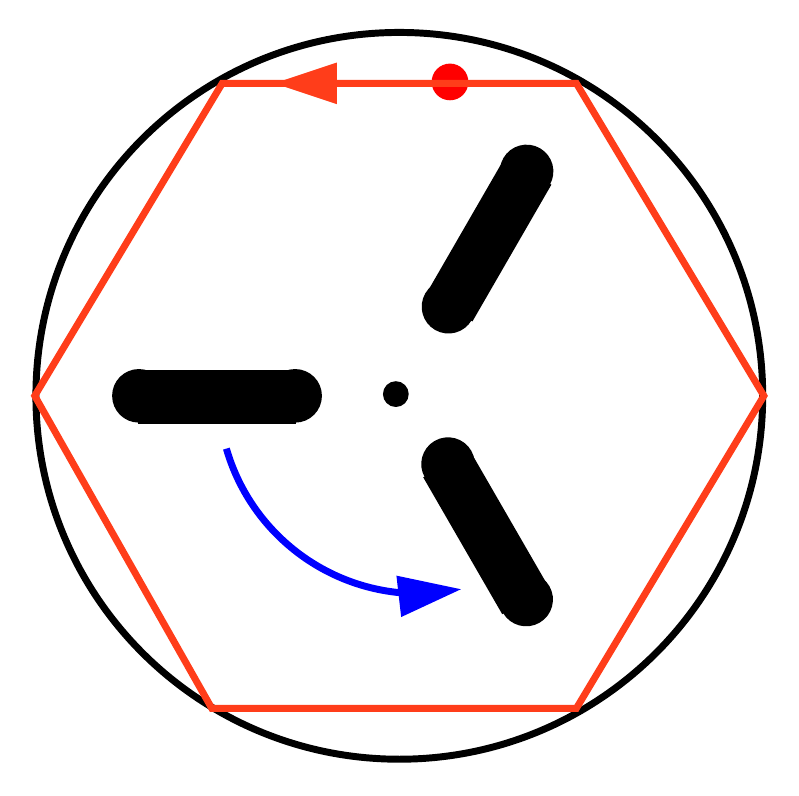}
  \captionof{figure}{An invariant measure is represented by the closed orbit of a single particle.}
  \label{fig3}
\end{figure}

Random reflections can be rigorously defined for hard balls of any size but
on the physics side, only reflections of small balls from a rough surface
can be expected to be random according to the definitions given below. 

We will limit our model to random reflections which are natural in two
ways: (i) the microcanonical  ensemble is an invariant measure for the
model with random reflections and (ii) the random reflections are the limit
of deterministic reflections from rough (fractal) surfaces, assuming that
the size of small crystals forming the rough surface goes to zero.   
Reflection laws with these properties have been studied in \cite{ABS,Feres07,PlakDdim,Plak2dim,Plakbook}.
We will  recall the characterization of such laws in Section \ref{a23.1}. We
will also state and prove an immediate corollary of that result, saying
that if the random angle of reflection does not depend on the angle of
incidence then the reflection law must be Lambertian, also known as the
Knudsen law, 
defined as follows. 
Let $S^{d-1}$ be the unit sphere in $\R^d$, let
$S^{d-1}_+ = \{(x_1,x_2,\dots, x_d)\in S^{d-1}: x_d> 0\}$, and $\bn = (0,0,\dots, 0,1)$. 
We say that the  probability measure $\nu_d$ on $S^{d-1}_+$ is Lambertian if  its density with respect to usual surface measure is
\begin{align}\label{f22.1}
f(v)=b_d \langle v,\bn\rangle,
\end{align}
where $b_d>0$ is the normalizing constant. 

Although Lambert's work \cite{L1760} precedes that of Knudsen \cite{K1934}
by more than one and a half centuries, it is appropriate to use Knudsen's
name in our context because he applied his model to gases, while Lambert
was concerned with reflections of light. 

We will consider a cylindrical drum rotating about its axis and a single 
pointlike particle reflecting from its surface. We will assume that the
reflections are Lambertian in the frame of reference $F$. 
We will derive a  formula for the expected value of the duration of the
free flight between reflections from the cylindrical walls. We will show
that the asymptotic winding number about the axis of rotation for the
particle is equal to the angular velocity of the cylindrical rotor. 

\subsection{Rotating drum with gravitation and Lambertian reflections} 

The last model that we consider includes Lambertian reflections and
gravitation. It is inspired by a bingo machine (see Fig. \ref{fig2old}). In
this case, we consider a  pointlike particle reflecting from an infinitely
heavy circular wall of a 2-dimensional drum with smooth or rough surface, and subject
to gravitational force acting in the plane of rotation. We argue that there
is Fermi acceleration in this model. The precise statement is given in
Section \ref{sec:grav}. We note here that the mathematical heart of the
argument is concerned with a simplified model in which two points are fixed
on the wall in the inertial frame of reference (so that they do not move
with the wall; see Fig. \ref{fig5}). We study a trajectory whose
reflections alternate between the two fixed points. The segments of the
trajectory become closer and closer to line segments and the energy of the
reflecting particle goes to infinity. This simplified model is the basis of
somewhat more realistic theorems.

\begin{figure}
\centering
\begin{minipage}{.45\textwidth}
  \centering
  \includegraphics[width=1\linewidth]{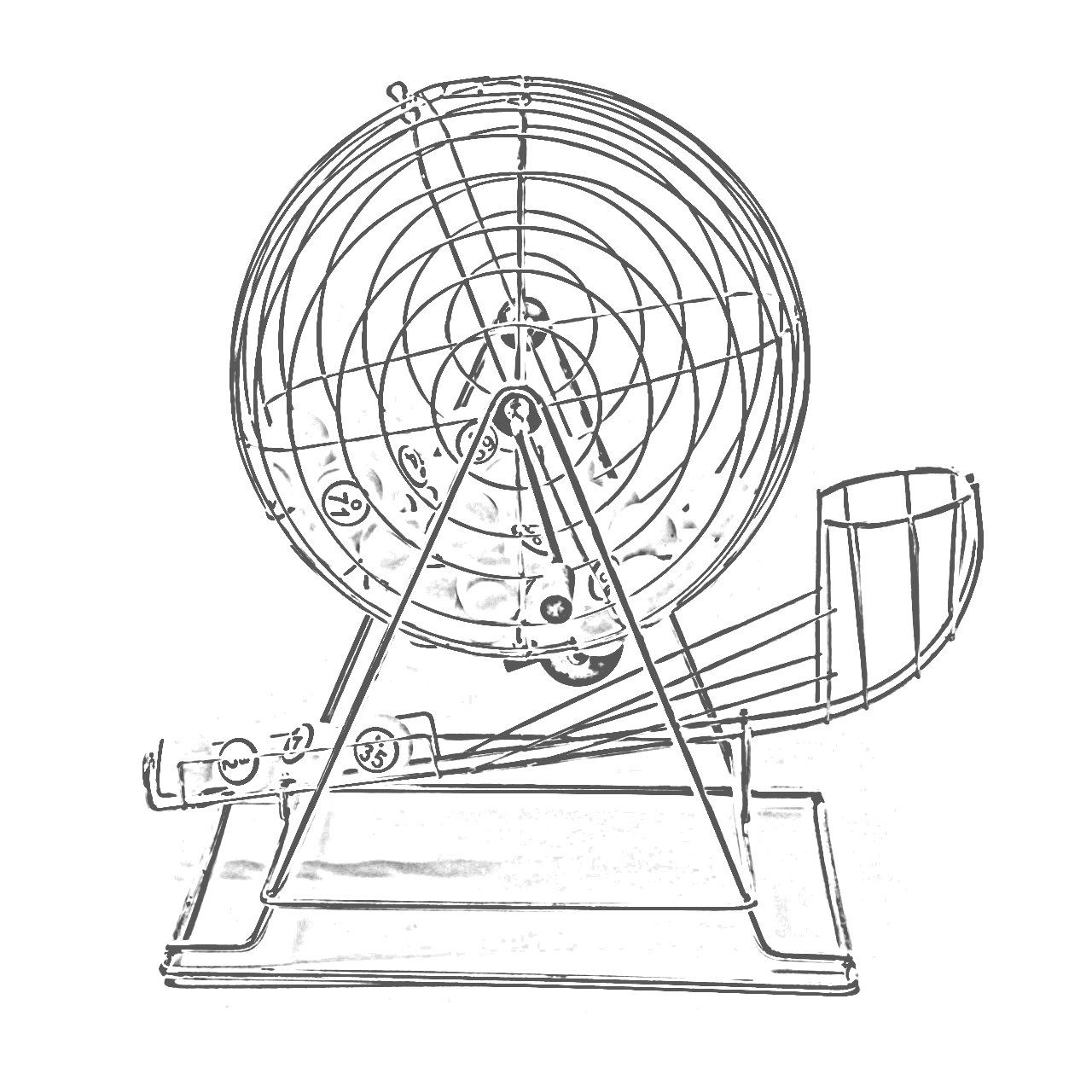}
  \captionof{figure}{Bingo machine.}
  \label{fig2old}
\end{minipage}%
\begin{minipage}{.45\textwidth}
  \centering
  \includegraphics[width=1\linewidth]{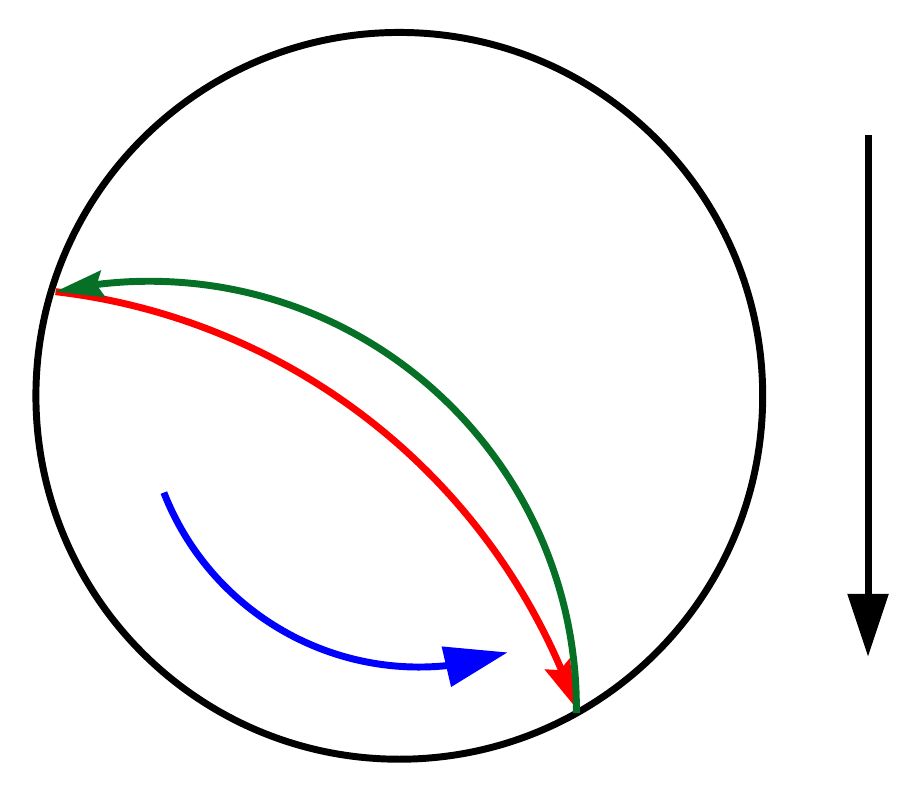}
  \captionof{figure}{Rotating drum with gravitation and pseudo-Lambertian
    reflections. The black arrow represents gravity. The blue arrow
    represents angular velocity. The red and green paths represent two (of
    infinitely many) parabolic segments of a trajectory hitting two points
    fixed in the inertial frame of reference.} 
  \label{fig5}
\end{minipage}
\end{figure}

\subsection{Literature review}

Our paper involves Fermi acceleration, billiards with gravity, billiards
with rotation, and Lambertian reflections, so we briefly review the
literature on  these topics. 

There is vast literature on Fermi acceleration. We point out a recent
addition \cite{Zhou2019}, concerned with rectangular billiards with moving
slits, because it is both interesting and contains an extensive literature
review. See \cite{Dol2008,GRKT,LLC,Pust95} for surveys of Fermi
acceleration. 

Billiards with gravitation or other forces or potentials were considered in 
\cite{DACOSTA2015,Dol,Dullin,Richter_1990}.

Billiards models incorporating  time-dependent boundaries, rotation,
breathing walls or similar effects originated in the papers
\cite{Fermi49,Ulam61} and were more recently analyzed in
\cite{AZ,DACOSTA2015,Dullin,Meyer95,Pust05}. 

The relationship between  rough billiard boundaries and random reflections
(including Lambertian reflections; equivalently Knudsen law) was studied in 
\cite{ABS,Feres16,Popov10,Feres12,Feres07,Feres13,Feres10,Feres12a,PlakDdim,Plak2dim,Plakbook}.

\section{Rotating drum with arbitrary shape and specular reflections}\label{sec:gen}

In this section we consider a rotating drum represented by a bounded domain
$D\subset \R^d$, for some $d\geq 2$, with smooth boundary, rotating with a
fixed angular velocity $\omega>0$.
The rotation applies only to a 2-dimensional plane $H$.  All 
other coordinates of points in $D$ remain constant. We suppose that the
drum $D$ holds $n$ hard (non-intersecting) balls for some finite $n\ge 1$.  
We will write $m_k >0$ for the mass of the $k$-th ball,  $v_k(t)\in \R^d$
will stand for its velocity at time $t$, and $x_k(t)\in \R^d$ will be  its
center at time $t$. 

To avoid uninteresting technical complications, we will assume that the
curvature of $\prt D$ is small relative to the radii of the balls. More
precisely, let $r_{\max}$ be the maximum of the ball radii. We will assume
that for every point $y \in \prt D$, there is exactly one open ball of
radius  $r_{\max}$ contained in $D$ whose boundary is tangent to $\prt D$
at $y$, and the closure of this ball touches $\partial D$ only at $y$.  

In this section we will assume that $\prt D$ and the surfaces of the balls
are perfectly smooth so that their collisions do not involve friction and,
therefore, they do not change the angular velocities of the balls. Hence,
we can and will assume that the balls do not rotate.  

We will assume that there are no simultaneous collisions of more than two
balls (see \cite{BurdzyDuarte20} for a precise formulation).  Unlike
for collisions of two balls, in a
simultaneous collision of more than two balls the evolution of the system
is not uniquely 
determined by conservation of energy, momentum and angular momentum (see,
for example \cite{Vaser79}).  
In view of this assumption and that of the previous paragraph, for the
evolution of the system to be uniquely 
defined at a collision, it is enough to assume  conservation of energy and
momentum.  We will also assume that when a ball collides with the wall
$\prt D$, it  
does not collide with another ball at the same time.  The law of 
specular reflection in the rotating frame then determines the evolution  
at collisions involving the wall.  
The wall $\prt D$ will be assumed to be infinitely heavy so
that we will exclude it from the energy and momentum balances.      

Let $F$ be the non-inertial frame of reference which makes  $D$ static.
For the $k$-th particle, let  $v^F_k(t)\in \R^d$ be its velocity in $F$ at
time $t$, given by \eqref{velocities}, and $x^H_k(t)\in H$ be the  
projection of its center on $H$ at time $t$.  Recall the definitions
\eqref{Fenergies} of the energy quantities in $F$.  

Recall from Section \ref{RD} that $L=\calR(\pi/2)\circ P_H$, and note that for all $x$,
\begin{align}\label{jy27.1}
\langle x, L(x)\rangle = \langle x - P_H(x), \calR(\pi/2)\circ P_H(x) \rangle + \langle P_H(x), \calR(\pi/2)\circ P_H(x) \rangle = 0 + 0 = 0.
\end{align}

\begin{proposition}\label{f17.1}
The energy
$E^F(t)$ is a conserved quantity, i.e., it depends only on the initial
conditions; it does not depend on time $t$. 
\end{proposition}

\begin{proof}
It follows from Proposition~\ref{rotatingprop} that $E^F(t)$ 
does not depend on $t$ between collisions. 
We will next argue that $E^F(t)$  does not change its value at collision times.

The positions $x_k(t)$ are continuous across a collision time while the
velocities $v_k(t)$ may have a jump discontinuity.  The quantity
$E^{F,P}(t)$ is therefore continuous so it suffices to show that  
$E^{F,K}(t)$ is continuous. 

Suppose that balls $j$ and $k$ collide at time $s>0$.  We denote 
$v_j^\pm =\lim_{t\to s^{\pm}}v_j(t)$  
and similarly for $v_k^\pm$, $v_j^{F,\pm}$, and $v_k^{F,\pm}$.  
Since the collision is elastic, momentum and energy are conserved across
the collision, so   
\begin{align*}
m_jv_j^+ + m_kv_k^+&= m_jv_j^- +m_kv_k^-\\
m_j\|v_j^+\|^2 + m_k\|v_k^+\|^2&= m_j\|v_j^-\|^2 
+m_k\|v_k^-\|^2.
\end{align*}
Using \eqref{velocities}, we have
\begin{align*}
m_j\|v_j^F(t)\|^2+m_j\|v_k^F(t)\|^2
&=m_j\|v_j(t)-\om L(x_j(t))\|^2
+m_k\|v_k(t)-\om L(x_k(t))\|^2\\
&=m_j\|v_j(t)\|^2+m_k\|v_k(t)\|^2 \\
&\quad +m_j\om^2\|L(x_j(t))\|^2+m_k\om^2\|L(x_k(t))\|^2\\ 
&\quad -2m_j\om\langle v_j(t),L(x_j(t))\rangle
-2m_k\om\langle v_k(t),L(x_k(t))\rangle.  
\end{align*}
The first line after the last equality is continuous across the collision
by conservation of energy in the inertial frame.  The second line after the 
last equality is continuous since the quantities which appear in it 
are continuous.  So it suffices to show that the last line is continuous
across the collision.  
Denote $x_j(s)$, $x_k(s)$ simply by $x_j$, $x_k$.  Using 
conservation of momemtum across the collision in the inertial frame, we
have  
\begin{align*}
\lim_{t\to s^+}&\big[m_j\langle v_j(t),L(x_j(t))\rangle
+m_k\langle v_k(t),L(x_k(t))\rangle\big]\\
&=m_j\langle v_j^+,L(x_j)\rangle +m_k\langle v_k^+,L(x_k)\rangle\\ 
&=m_j\langle v_j^+,L(x_j)\rangle +m_k\langle v_k^+,L(x_j)\rangle  
+m_k\langle v_k^+,L(x_k-x_j)\rangle\\
&=\langle m_jv_j^++m_kv_k^+,L(x_j)\rangle 
+m_k\langle v_k^+,L(x_k-x_j)\rangle\\
&=\langle m_jv_j^-+m_kv_k^-, L(x_j)\rangle
+m_k\langle v_k^+,L(x_k-x_j)\rangle\\
&=m_j\langle v_j^-,L(x_j)\rangle +m_k\langle v_k^-,L(x_k)\rangle 
+m_k\langle v_k^-,L(x_j-x_k)\rangle
+m_k\langle v_k^+,L(x_k-x_j)\rangle\\ 
&=\lim_{t\to s^-}\big[m_j\langle
  v_j(t),L(x_j(t))\rangle 
+m_k\langle v_k(t),L(x_k(t))\rangle\big]
+m_k\langle v_k^+-v_k^-,L(x_k-x_j)\rangle.   
\end{align*}
In an elastic collision of two balls, the component of velocity of each
ball orthogonal to the line connecting the centers at the
time of collision is continuous; it is only the 
component of velocity in the direction of the line of centers which
reflects.  So $v_k^+-v_k^-$ is a multiple of $x_k-x_j$.  Since
$\langle x,L(x)\rangle=0$ for all $x$ (see \eqref{jy27.1}), it follows the second term on the
last line vanishes so that $E^{F,K}(t)$ is continuous across a collision
time of two balls. 

It is clear that $E^F(t)$ is continuous at the time of a collision 
between a ball and the wall $\prt D$ since the ball reflects according to 
the law of specular reflection in the frame $F$, so that 
$\|v^F(t)\|$ is preserved.  We conclude that $E^F(t)$ is in fact a constant
$E^F$.  
\end{proof}

\begin{proposition}
Let 
\begin{align*}
R =\sup_{x\in D}\|x^H\|, \quad M=\sum_{k=1}^n  m_k, 
\quad E^{K}(t)=\sum_{k=1}^n  m_k \|v_k(t)\|^2/2,
\end{align*}
and let $E^F:=E^F(0)=E^F(t)$.  For all $t>0$,  
\[
E^K(t)\leq 2E^F +2M\omega^2R^2.
\]
In particular, since $E^F$ is constant, there is no Fermi acceleration. 
\end{proposition}

\begin{proof}
For all $t$,
\begin{align}\label{f13.8}
0 \leq \left| E^{F,P}(t) \right|
=  \sum_{k=1}^n  m_k \omega^2 \|x^H_k(t)\|^2/2
\leq M\omega^2 R^2/2 .
\end{align}
We have 
$\|\calR(\omega t)(v_k^F(t)) -v_k(t)\|= \omega\|L(x_k(t))\|
\leq \omega R$, so 
\begin{align*}
\| v_k(t)\|^2 &\leq 2\|v_k^F(t) \|^2+2\|\calR(\omega t)(v_k^F(t))-v_k(t)\|^2    
\leq 2\|v_k^F(t) \|^2 +  2 \omega^2 R^2.
\end{align*}
This and \eqref{f13.8} imply that for any $t>0$,
\begin{align*}
E^K(t)&=\sum_{k=1}^n \frac{m_k}2 \|v_k(t)\|^2
\leq \sum_{k=1}^n m_k \left(\|v_k^F(t) \|^2 +   \omega^2 R^2\right)
= 2 E^{F,K}(t)+M  \omega^2 R^2 \\
&=2(E^F- E^{F,P}(t))+M\omega^2R^2
\leq 2E^F+ 2M\omega^2 R^2. 
\end{align*}
\end{proof}

The next proposition concerns the invariant measure.  
Proposition~\ref{rotatingprop} describes the invariant measure for the
dynamical system consisting of noninteracting point particles free to roam 
in the whole Euclidean space.  Our vector of centers 
$\left(x_1^F,x_2^F,\dots,x_n^F \right)$ is constrained to lie in $\cD$,
which has boundaries corresponding to collisions of the balls or of balls 
and the wall $\partial D$.  Recall the level set $\cS_{E^F}$ and its
subsets $\cS_{E^F,\cD}$ and $\cS^0_{E^F,\cD}$ defined in the introduction.   
\begin{proposition}\label{f22.6}
Fix energy $E^F$ such that $\cS^0_{E^F,\cD}\neq \emptyset$.  The measure      
defined in \eqref{f13.10} and \eqref{f13.11} is the restriction to
$\cS^0_{E^F,\cD}$ of an invariant measure on $\cS_{E^F,\cD}$.   
\end{proposition}

\begin{proof}
The invariant measure on $\cS_{E^F}$ in the statement of 
Proposition~\ref{rotatingprop} restricts to a measure 
on $\cS_{E^F,\cD}\subset \cS_{E^F}$ (making the obvious adjustment since 
the discussion in Section~\ref{ME} is formulated in terms of momenta
instead of velocities), which is invariant for  
$\left(x_1^F,\dots,x_n^F\right)\notin \partial \cD$, and whose restriction 
to $\cS^0_{E^F,\cD}$ has the form \eqref{f13.10}, \eqref{f13.11}.  
It remains to analyze the collisions.    

Recall that we are not considering simultaneous collisions of more than two
balls nor collisions in which a ball collides with $\partial D$ at the same
time that it collides with another ball.  It is known that the set of
initial conditions giving rise to simultaneous collisions of more  
than two balls has measure zero (\cite{Alex}).  We believe that the
arguments of that paper apply also to simultaneous collisions of 
$\partial D$ and more than one ball.  

At each point of $\partial \cD$ corresponding to a collision of two balls
or of a ball and $\partial D$, there is a collision map from the space 
of vectors of incoming velocities to that of outgoing velocities.  It
suffices to show that this map on velocities preserves the   
measure $d\sigma_1\left(\vb_1^F,\vb_2^F,\ldots,\vb_n^F\right)$ appearing in 
\eqref{f13.10}.    

Suppose that balls $j$ and $k$ collide at time $s>0$.  The components of
$v_j(t)$ and $v_k(t)$ orthogonal to the line connecting the locations 
$x_j(s)$ and $x_k(s)$ of the centers at time $s$ are continuous; only 
the components parallel to this line can jump.  Set 
$\bm{e^\|}=\frac{x_k(s)-x_j(s)}{\|x_k(s)-x_j(s)\|}$ and 
\[
v_j^{\|\pm}=\lim_{t\to s^\pm}\left\langle v_j(t),\bm{e^\|}\right\rangle,\qquad
v_k^{\|\pm}=\lim_{t\to s^\pm}\left\langle v_k(t),\bm{e^\|}\right\rangle.
\]
A collision occurs only if 
$v_j^{\|-} > v_k^{\|-}$.  This condition defines the space of incoming
velocities.  The space of outgoing velocities is defined by the
complementary condition $v_j^{\|+} < v_k^{\|+}$.   
For the frame $F$ we similarly set 
\[
v_j^{F,\|\pm}=\lim_{t\to s^\pm}\left\langle\mathcal{R}(\omega t)
v_j^F(t),\bm{e^\|}\right\rangle,\qquad
v_k^{F,\|\pm}=\lim_{t\to s^\pm}\left\langle\mathcal{R}(\omega t)
v_k^F(t),\bm{e^\|}\right\rangle.
\]
In light of \eqref{velocities}, we have 
\begin{equation}\label{velocitiespm}
v_j^{F,\|\pm}=v_j^{\|\pm}-\omega 
\left\langle L(x_j(s)),\bm{e^\|}\right\rangle
\end{equation}
and similarly for $v_k^{F,\|\pm}$.  On the right-hand side, the term 
$-\omega \left\langle L(x_j(s)),\bm{e^\|}\right\rangle$ is the same for 
both equations corresponding to $\pm$.  
Since $x_j(s)-x_k(s)$ is a multiple of $\bm{e^\|}$ and 
$\langle L(x),x\rangle = 0$ for all $x$, the
incoming and outgoing conditions are equivalent to 
$v_j^{F,\|-} > v_k^{F,\|-}$ and $v_j^{F,\|+} < v_k^{F,\|+}$.  

Standard formulas for the transformation of velocities in a  totally
elastic reflection of balls in one dimension can be written in the
following form
\[
\begin{pmatrix}
\sqrt{m_j}v_j^{\|+} \\
\\
\sqrt{m_k}v_k^{\|+}
\end{pmatrix}
=
\begin{pmatrix}
\displaystyle\frac{2\sqrt{m_jm_k}}{m_j+m_k} 
& \displaystyle\frac{m_j-m_k}{m_j+m_k} \\ 
\\
\displaystyle\frac{m_k-m_j}{m_j+m_k} 
& \displaystyle\frac{2\sqrt{m_jm_k}}{m_j+m_k} 
\end{pmatrix}
\begin{pmatrix}
\sqrt{m_k}v_k^{\|-} \\ 
\\
\sqrt{m_j}v_j^{\|-}  
\end{pmatrix}.
\]
We claim that the same relation holds with 
$v_j^{\|+}$, $v_k^{\|+}$, $v_j^{\|-}$, $v_k^{\|-}$ replaced by
$v_j^{F,\|+}$, $v_k^{F,\|+}$, $v_j^{F,\|-}$, $v_k^{F,\|-}$, i.e. 
\begin{equation}\label{Frelation}
\begin{pmatrix}
\sqrt{m_j}v_j^{F,\|+} \\
\\
\sqrt{m_k}v_k^{F,\|+}
\end{pmatrix}
=
\begin{pmatrix}
\displaystyle\frac{2\sqrt{m_jm_k}}{m_j+m_k} 
& \displaystyle\frac{m_j-m_k}{m_j+m_k} \\ 
\\
\displaystyle\frac{m_k-m_j}{m_j+m_k} 
& \displaystyle\frac{2\sqrt{m_jm_k}}{m_j+m_k} 
\end{pmatrix}
\begin{pmatrix}
\sqrt{m_k}v_k^{F,\|-} \\ 
\\
\sqrt{m_j}v_j^{F,\|-}   
\end{pmatrix}.
\end{equation}
According to \eqref{velocitiespm}, this is equivalent to the identity    
\[
\begin{pmatrix}
\sqrt{m_j}\left\langle L(x_j(s),\bm{e^\|}\right\rangle \\
\\
\sqrt{m_k}\left\langle L(x_k(s),\bm{e^\|}\right\rangle 
\end{pmatrix}
=
\begin{pmatrix}
\displaystyle
\frac{2\sqrt{m_jm_k}}{m_j+m_k} & 
\displaystyle\frac{m_j-m_k}{m_j+m_k} \\
\\
\displaystyle\frac{m_k-m_j}{m_j+m_k} & 
\displaystyle\frac{2\sqrt{m_jm_k}}{m_j+m_k} 
\end{pmatrix}
\begin{pmatrix}
\sqrt{m_k}\left\langle L(x_k(s),\bm{e^\|}\right\rangle \\
\\
\sqrt{m_j}\left\langle L(x_j(s),\bm{e^\|}\right\rangle 
\end{pmatrix}.
\]
A little computation shows that both components of this vector equation
reduce to $\left\langle L(x_j(s)-x_k(s)),\bm{e^\|}\right\rangle =0$.   
As above, this holds since $x_j(s)-x_k(s)$ is a multiple of $\bm{e^\|}$ and  
$\langle L(x),x\rangle=0$ for all $x$.  Thus \eqref{Frelation} is verified.   

Since $m_j,m_k>0$ and
\begin{align*}
\left(\frac{2\sqrt{m_j m_k}}{m_j+m_k}\right)^2 
+ \left(\frac{m_j -m_k}{m_j+m_k}\right)^2 =1,
\end{align*}
there is a unique $\alpha\in(-\pi/2, \pi/2)$ such that
\begin{align*}
\cos \alpha = \frac{2\sqrt{m_j m_k}}{m_j+m_k},
\qquad \sin \alpha = \frac{m_k -m_j}{m_j+m_k}.
\end{align*}
Hence we can write
\begin{align}\label{f17.2}
\begin{pmatrix}
\displaystyle\sqrt{m_j}v_j^{F,\|+} \\
\displaystyle\sqrt{m_k}v_k^{F,\|+}
\end{pmatrix}
=
\begin{pmatrix}
\cos \alpha  & 
-\sin \alpha \\
\sin \alpha & 
\cos \alpha 
\end{pmatrix}
\begin{pmatrix}
\displaystyle\sqrt{m_k}v_k^{F,\|-} \\
\displaystyle\sqrt{m_j}v_j^{F,\|-}
\end{pmatrix}.
\end{align}
The incoming condition $v_j^{F,\|-} > v_k^{F,\|-}$ can be rewritten as 
\begin{align}\label{f18.1}
\sqrt{m_j}v_j^{F,\|-} > \sqrt{\frac{m_j}{m_k}} 
\sqrt{m_k}v_k^{F,\|-}.
\end{align}
Straightforward calculations show that the map
\[
\begin{pmatrix}
u_1^-\\
u_2^-
\end{pmatrix}
\mapsto
\begin{pmatrix}
u_1^+\\
u_2^+
\end{pmatrix}
=\begin{pmatrix}
\cos \alpha  & 
-\sin \alpha \\
\sin \alpha & 
\cos \alpha 
\end{pmatrix}
\begin{pmatrix}
u_2^-\\
u_1^-
\end{pmatrix}
\]
is a bijection between the incoming half-plane   
$u_1^->\sqrt{\frac{m_j}{m_k}} u_2^-$ and the outgoing half-plane   
$u_1^+<\sqrt{\frac{m_j}{m_k}} u_2^+$.

The transformation in \eqref{f17.2} is the composition of the symmetry
\begin{align*}
\left(\sqrt{m_j}v_j^{F,\|-} ,
\sqrt{m_k}v_k^{F,\|-} \right)
\to
\left( \sqrt{m_k}v_k^{F,\|-} ,
\sqrt{m_j}v_j^{F,\|-} \right) 
\end{align*}
and rotation by the angle $\alpha$. 
Since the collision map leaves constant the velocity components orthogonal
to $\bm{e^\|}$ and the velocities of the noncolliding balls, it is an
orthogonal transformation in the space of 
$\left(\vb_1^F,\vb_2^F,\ldots,\vb_n^F\right)$ given by \eqref{vbar}, so 
certainly it maps the measure 
$d\sigma_1\left(\vb_1^F,\vb_2^F,\ldots,\vb_n^F\right)$ on the incoming
half-sphere to the same measure on the outgoing half-sphere.  
This completes the proof in the case of the collision of two balls.  

The collision of a ball with the wall of the rotating drum is simpler.  
If the $j$-th ball collides with $\partial D$ at time $s$ and $\bm{\nu}$
denotes the outward-pointing normal at the point of contact, the incoming 
velocities are those for which $\langle v_j^F(s),\bm{\nu}\rangle>0$
and the outgoing velocities are those for which 
$\langle v_j^F(s),\bm{\nu}\rangle<0$.  The collision map just reflects
$v_j^F(s)$ across the  
tangent plane to $\partial D$ and leaves all other
velocities unchanged.  This map clearly preserves the measure 
$d\sigma_1\left(\vb_1^F,\vb_2^F,\ldots,\vb_n^F\right)$.  
\end{proof}

\section{A pointlike particle in rotating drum}

In this section, 
we will focus on the case of a single pointlike particle in a rotating
drum. This model is inspired by ``Knudsen gas,'' i.e.,  gas diluted to the
point that molecules typically do not collide on the length scale of the
diameter of the drum; see \cite{K1934}.  

\subsection{Lambertian reflections as a model for a rough surface}\label{a23.1}

In this subsection we will prove in Corollary \ref{f21.6} that the Knudsen reflection law \eqref{f22.1} is the
only random reflection distribution which can arise from classical specular
reflections from a fractal surface consisting of small smooth crystals if
we assume homogeneity and lack of memory, i.e., if we assume that the
limiting distribution of the angle of reflection does not depend on the
angle of incidence, and it is the same at all boundary points.
 
Corollary \ref{f21.6} follows easily from results in
 \cite{ABS,Feres07,PlakDdim,Plak2dim,Plakbook}, stated below as Propositions \ref{j7.1} and \ref{m16.2}.
The results in \cite{ABS} rediscovered (independently) those in
\cite{Plak2dim}; both articles were concerned with two-dimensional
models. The full generality in an arbitrary number of dimensions was
achieved in  \cite{PlakDdim, Plakbook}. 
Our presentation will combine  \cite[Thms. 2.2, 2.3]{ABS},
\cite[Sect.~4]{Feres07} and \cite[Thms. 4.4, 4.5]{Plakbook}. 
For a very detailed and careful presentation of the billiards
model in the plane see  \cite[Ch. 2]{CM}. 
For the multidimensional setup see \cite{Plakbook}. We will be somewhat
informal as the technical aspects of the model are tangential to our
project. 

Suppose that $M$ is a subset of the half-space 
$\{(x_1,x_2,\dots, x_d)\in \R^d: x_d\leq 0\}$
and its intersection with every ball (with finite radius) consists of a
finite number of smooth surfaces. 
Informally, $M$ consists of very small mirrors (walls of billiard tables)
that are supposed to model a macroscopically flat but rough reflecting
surface. 

Recall \eqref{f22.1} and the definitions preceding it.
Let 
$L_* =\{(x_1,x_2,\dots, x_d)\in \R^d: x_d= 0\}$
and $B= L_* \times S^{d-1}_+ $.
Define a $\sigma$-finite
measure $\Lambda(dx, dv)$ on $B$ as the product of Lebesgue measure on
$L_*$ and  $\nu_d$ on $S^{d-1}_+$. 

We will consider
a pointlike particle trajectory in $\R^d$ starting on $L_*$.
We will consider trajectories which are straight line segments between
reflections and will assume that they reflect from surfaces
comprising $M$ 
according to the rule of specular reflection, that is, the
angle of incidence is equal to the angle of reflection, for every reflection.

Suppose that a trajectory starts from  $x \in L_*$ in the direction $-v$,
where $v\in S^{d-1}_+$, at time 0, reflects from 
surfaces of $M$ and returns to $y \in L_*$ 
at a time $t$, and
$t>0$ is the smallest time with this property. 
Let $w\in S^{d-1}_+$ be the velocity of the particle at time $t$. This defines a
mapping $K : B \to B$, given by $K(x,v) = (y,w)$.
Clearly, $K$ depends on $M$.

It can be proved (see, e.g., \cite[Prop. 2.1]{ABS}) that under mild and natural assumptions,
for $\Lambda$-almost all $(x,v)\in B$, a trajectory starting from
$x$ with velocity $-v$, and reflecting from surfaces comprising $M$ will return to $L_*$
after a finite number of reflections. 

We will write $\P(x, v; dy, dw)$ to denote a Markov
transition kernel on $B$, that is, for fixed $(x,v) \in B$,
$\P(x, v; dy, dw)$ is a probability measure on $B$. We
assume that $\P$ satisfies the usual measurability conditions
in all variables.

We will use $\delta_x(y)$ to denote Dirac's delta function.
Recall the transformation $K$ and let $\P_{K}$ be defined by
$\P_{K}(x, v; dy, dw) = \delta_{K(x,v)}(y,w)dydw$. In other
words, $\P_{K}$ represents a deterministic Markov kernel, with
the atom at $K(x,v)$.

If $\mu_n$, $n\geq 1$, and $\mu_\infty$ are non-negative
$\sigma$-finite measures on some measurable space $\Gamma$ then
we will say that $\mu_n$ converge weakly to $\mu_\infty$ if
there exists a sequence of sets $\Gamma_j$, $j\geq 1$, such
that $\bigcup_{j\geq 1} \Gamma_j = \Gamma$, $\mu_n(\Gamma_j) <
\infty$, $\mu_\infty(\Gamma_j)<\infty$ for all $n$ and $j$, and
for every fixed $j$, the finite measures $\mu_n|_{\Gamma_j}$ 
converge weakly to $\mu_\infty|_{\Gamma_j}$.   

\begin{proposition}\label{j7.1}
 Suppose that for some sequence of sets $M_n$,
corresponding transformations $K_n$, and some Markov transition kernel $
\P(x, v; dy, dw)$, we have 
\begin{align*}
\Lambda(dx, dv)  \P_{K_n} (x, v; dy, dw)
\to
\Lambda(dx, dv)  \P(x, v; dy, dw)
\end{align*}
in the sense of weak convergence
on $B^2$ as $n\to \infty$. Then $\P$ is symmetric
with respect to $\Lambda$ in the sense that for any smooth
functions $f$ and $g$ on $B$ with compact support we have
\begin{align*}
\int_{B^2} f(y,w) \P (x, v; dy, dw)
g(x,v) \Lambda(dx, dv)
= \int_{B^2} g(y,w) \P (x, v; dy, dw)
f(x,v) \Lambda(dx, dv).
\end{align*}
In particular, $\Lambda$ is invariant in the sense that
\begin{align}\label{f20.1}
\int_{B^2} f(y,w) \P (x, v; dy, dw) \Lambda(dx, dv)
= \int_{B} f(x,v) \Lambda(dx, dv).
\end{align}

\end{proposition}

Recall that $\delta_x(y)$ denotes Dirac's delta function.
Suppose that the probability kernel $\P$ in Proposition \ref{j7.1} 
satisfies $\P(x,v; dy, dw) = \delta_x(y)dy \wt \P(x,v; dw)$ for some $\wt \P$.
Heuristically, this means that the trajectory starting at $x$ is
instantaneously reflected from a mirror located infinitesimally close to
$L_*$.  Then  \eqref{f20.1} implies that for all smooth bounded functions
$f$ on $S^{d-1}_+$, and almost all $x$, 
\begin{align}\label{f21.2}
\int_{\left(S^{d-1}_+\right)^2} f(w) \wt \P (x, v; dw) \nu_d(dv)
= \int_{S^{d-1}_+} f(v) \nu_d(dv).
\end{align}
If, in addition, we assume that $\wt \P (x, v; dw)=\wt \P ( dw)$,
i.e., $\wt \P (x, v; dw)$ does not depend on $x$ and $v$, then 
for all smooth bounded functions $f$  on $S^{d-1}_+$, 
\begin{align*}
\int_{\left(S^{d-1}_+\right)^2} f(w) \wt \P ( dw) \nu_d(dv)
= \int_{S^{d-1}_+} f(v) \nu_d(dv),
\end{align*}
and, therefore,
\begin{align}\label{f21.4}
\int_{S^{d-1}_+} f(w) \wt \P ( dw) 
= \int_{S^{d-1}_+} f(v)  \nu_d(dv).
\end{align}
Hence, $\wt \P ( dw)= \nu_d(dw)$.  
We have just proved the following result.

\begin{corollary}\label{f21.6}

Suppose that for some sequence of sets $M_n$,
corresponding transformations $K_n$, and some  $\wt \P( dw)$, we have
\begin{align*}
\Lambda(dx, dv)  \P_{K_n} (x, v; dy, dw)
\to
\Lambda(dx, dv) \wt \P( dw)
\end{align*}
in the sense of weak convergence
on $B^2$ as $n\to \infty$. Then 
\begin{align}\label{f22.2}
\wt \P ( dw) = \nu_d( dw), \qquad w\in S^{d-1}_+.
\end{align}
\end{corollary}

Recall that a reflection law is called Lambertian or Knudsen  if it is
given by \eqref{f22.1} (in particular, it does not depend on the angle of
incidence). Hence, under the assumptions of Corollary \ref{f21.6}, the
limiting reflection law is Lambertian. 

\begin{proposition}\label{m16.2}
Suppose that $\P(x,v; dy, dw) = \delta_x(y)dy \wt \P(x,v; dw)$
where $\wt \P$ satisfies for smooth $f$ and $g$, 
\begin{align*}
\int_{\left(S^{d-1}_+\right)^2} f(w) \wt \P (x, v; dw)
g(v) \nu_d(dv)
= \int_{\left(S^{d-1}_+\right)^2} g(w) \wt\P (x, v; dw)
f(v) \nu_d(dv).
\end{align*}
Then there exists a sequence of sets $M_n$ and corresponding
transformations $K_n$ such that
\begin{align*}
\Lambda(dx, dv)  \P_{K_n} (x, v; dy, dw)
\to
\Lambda(dx, dv)  \P(x, v; dy, dw)
\end{align*}
weakly on $B^2$ as $n\to \infty$. Moreover, $M_n$ can be chosen
in such a way that $M_n \subset \{(x_1,x_2,\dots,x_d) : -1/n < x_d \leq 0\}$.

\end{proposition}

\begin{remark}\label{ja19.1}
The proposition applies to $\wt \P$ in \eqref{f22.2}, so a dynamical system
with Lambertian reflections is the limit of a sequence of systems with
specular deterministic reflections. 

We will later apply our results from Section \ref{ME} and Proposition
\ref{f22.6} on the stationary distribution given in
\eqref{f13.10}-\eqref{f13.11} to the approximating sets $M_n$. Since we
prove these results for smooth boundaries, we point out that it is easy to
see that the sets $M_n$ in Proposition \ref{m16.2} can be chosen to be
smooth. 

\end{remark}

\subsection{Pointlike particle in cylindrical drum}

We will now compute the average flight time for a single particle in a
rotating cylindrical drum in $\R^d$, $d\geq 2$, with a rough surface, i.e., with the Knudsen
reflection law. 

Since we are dealing with a single particle in this section, we will drop
the subscript from the notation introduced in Section \ref{RD} and we will
write $x(t)$ instead of $x_1(t)$, $m$ instead of $m_1$, etc. 
We will assume that the drum  is the product of a finite ball and a finite or infinite cube, 
i.e., 
\begin{align*}
D_{\ell}=\left\{(z_1, z_2, \dots, z_d) : z_1^2 + z_2^2 \leq \rho^2, |z_k|
\leq \ell \text {  for  } 3\leq k \leq d\right\}, 
\end{align*}
for some $\rho \in (0,\infty)$ and $\ell \in (0,\infty)$, and 
$D_\infty = \cup_{l\geq 0}D_\ell$.  
Let
\begin{align*}
\prt _c D_{\ell}=\left\{(z_1, z_2, \dots, z_d)\in \prt D_\ell : z_1^2 + z_2^2 = \rho^2\right\},
\qquad \ell\in (0,\infty].
\end{align*}
We will also need the torus $D_{ T}$ defined by taking $D_{\infty}$ and
identifying any two points  $(z_1, z_2, \dots, z_d)$ and $(z'_1, z'_2,
\dots, z'_d)$ if for every $k=3,\dots,d$, $z_k - z'_k$ is an even
integer. The part of the boundary  
$\prt _c D_{T}$ is defined in a way analogous to the definition of $\prt _c D_{\ell}$.

Recall that for $z=(z_1, z_2, \dots, z_d)\in \R^d$, $z^H$ denotes $(z_1,z_2)$.

Suppose that the particle starts in $D_{\ell}$ at time $t=0$ with non-zero
velocity, where $\ell \in (0,\infty]$. 
Let $\tau_\ell$ be the time of the first
reflection from the boundary.  The time $\tau_T$
is defined in an analogous way for the particle flight in the torus $D_{ T}$.

Let $\E$ be the expectation corresponding to starting from a point in
$\prt_c D_\infty$ with direction of velocity in the rotating frame     
distributed according to the Knudson law \eqref{f22.2}.  We will be  
concerned only with $\E \tau_\infty$ which does not depend on the exact
location of the starting point in $\prt_c D_\infty$, by rotational
symmetry and shift invariance, so the starting point is not included in the
notation.   

Let  $\calA(d,r)$ denote the surface area of a $d$-dimensional
sphere with radius $r$ (i.e., the boundary of a $(d+1)$-dimensional ball
with radius $r$). By convention,  $\calA(0,r) = 2$.  Let 
$\calB(d,1)=\calA(d-1,1)/d$ denote the Lebesgue measure of the unit ball in
$\mathbb{R}^d$.  

\begin{theorem}\label{f22.3}

Consider a single pointlike particle reflecting from the walls of $D_\ell$
according to the Knudsen law. Fix any $E^F$. If $E^F <0$ then 
let $\rho_0$ be defined by 
$E^F +   m \omega^2 \rho_0^2/2 =0$ so 
$\rho_0 = (- 2E^F/ (m\omega^2))^{1/2}$.  If $E^F\geq 0$, set $\rho_0=0$.  

(i)
If $0< \ell < \infty$ and $\rho_0 < \rho$ then
 the process  $(x^F(t), v^F(t)/\|v^F(t)\|)$  has a unique stationary
 distribution in $D_\ell \times S$, where $S$ is the $(d-1)$-dimensional
 unit sphere. 

For $0< \ell < \infty$, let $h_\ell(z,u)$ be the density of the stationary
measure with respect to the product of Lebesgue measure on $D_\ell$ with
normalized surface measure on $S$.   

(ii) If $0< \ell < \infty$ and $E^F\geq 0$ then for $z\in D_\ell$ and $u\in S$,
\begin{align}\label{f22.4}
h_\ell\left(z,u\right)
&=
\frac{d m \omega^2 }{4(2\ell)^{d-2} \pi}
\cdot
 \frac{
\left(E^F +   m \omega^2 \|z^H \|^2/2\right)^{d/2-1}}
{
 \left(E^F + m \omega^2 \rho^2/2 \right)^{d/2}
 - \left(E^F  \right)^{d/2}}. 
\end{align}

(iii) If  $E^F\geq 0$ then
\begin{align}
\label{f22.7}
 \E \tau_\infty &=
\sqrt{\frac{2}{m}}
 \cdot
\frac{\calB(d,1) }{\calB(d-1,1)}
\cdot \frac1 {\omega^{2} \rho }
\cdot
\frac
{\left( E^F + m \omega^2 \rho^2/2 \right)^{d/2}
 - \left( E^F  \right)^{d/2}}
{ \left(E^F +   m \omega^2 \rho^2/2\right)^{(d-1)/2}}.
\end{align}

(iv) If $0< \ell < \infty$, $E^F< 0$ and $\rho_0<\rho$
then for $z\in D_\ell$ and $u\in S$,
\begin{align}\label{f22.5}
h_\ell\left(z,u\right)
&=\frac{d m \omega^2 }{4(2\ell)^{d-2} \pi}
\cdot
 \frac{
\left(E^F +   m \omega^2 \|z^H \|^2/2\right)^{d/2-1}}
{
 \left( E^F + m \omega^2 \rho^2/2 \right)^{d/2} }
\bone_{(\rho_0,\rho)}(\|z^H \|).
\end{align}

(v) If $E^F< 0$ and $\rho_0<\rho$ then
\begin{align}
\label{f22.8}
 \E \tau_\infty &=
\sqrt{\frac{2}{m}}
\cdot
\frac{\calB(d,1) }{\calB(d-1,1)}
\cdot \frac1 {\omega^{2} \rho }
\cdot
\left( E^F + m \omega^2 \rho^2/2 \right)^{1/2}
=
\frac{\calB(d,1) }{\calB(d-1,1)}
\cdot \frac{v_*}{\omega^2\rho},
\end{align}
where 
\begin{align}\label{fe25.4}
v_*:=
\sqrt{(2/m)E^{F,K}} = \left((2/m)E^F +   \omega^2 \rho^2\right)^{1/2}
\end{align}
is the maximum speed in $F$.
\end{theorem}

\begin{proof}

(i)
The uniqueness follows from the following standard coupling argument.
Two copies of the process can be constructed so that they hit the same
point on the boundary at the end of their first flights (not necessarily
the same time for both processes), with positive probability. By the strong
Markov property, they will visit the same point on the boundary with
probability 1. Another application of the strong Markov property and
coupling shows that their evolutions following the visit of the same point
can be identical, up to a shift in time. The ergodic theorem then implies
that the stationary distributions have to be the same. 

According to Remark \ref{ja19.1}, there exists a sequence of sets
$D^k_\ell$, $k\geq 1$, converging to $D_\ell$ and such that the reflection
laws in  
$D^k_\ell$'s converge to the Knudsen reflection law in the sense of  
 Proposition \ref{m16.2}.
It follows from Proposition \ref{f22.6} that there exists
 a sequence of processes with specular reflection laws in $D^k_\ell$,
 $k\geq 1$, in the stationary regimes, with the stationary laws given by
 \eqref{f13.10} and \eqref{f13.11}. Taking the limit, we conclude that
 there is a process in $D_\ell$ with the Knudsen reflection law and the
 stationary distribution given by \eqref{f13.10} and \eqref{f13.11}. 
Hence, the form of the density $h_\ell(z,u)$ in \eqref{f22.4} and
\eqref{f22.5} is given by \eqref{f13.10} and \eqref{f13.11}, up to a
constant. We will find the normalizing constants. 

(ii)
Let $z^\perp = z - z^H$, $\calB_2(0,\rho)=\{z^H: \|z^H\| \leq \rho\}$ and
$\calQ_{d-2}(\ell)=\{z^\perp: |z_k| \leq \ell \text {  for  } 3\leq k \leq d\}$.  
The integral of unnormalized invariant measure is equal to
\begin{align}\label{fe25.1}
C_1 &:=\int _{\calQ_{d-2}(\ell)}\int _{\calB_2(0,\rho)}
 \left(E^F +   m \omega^2 \|z^H \|^2/2\right)^{d/2-1} dz^H 
 dz^\perp \\
 & = 
 (2\ell)^{d-2} \int _{\calB_2(0,\rho)}
 \left(E^F +   m \omega^2 \|z^H \|^2/2\right)^{d/2-1} dz^H \notag \\
 &  = 
 (2\ell)^{d-2} \int _0^\rho
 2\pi r
 \left(E^F +   m \omega^2 r^2/2\right)^{d/2-1} dr\notag \\
 &  = 
 (2\ell)^{d-2} \int _0^{\rho^2}
 \pi 
 \left(E^F +   m \omega^2 s/2\right)^{d/2-1} ds\notag \\
&  = 
 \frac{4(2\ell)^{d-2} \pi}{d m \omega^2 }\left(
 \left( E^F + m \omega^2 \rho^2/2 \right)^{d/2}
 - \left( E^F  \right)^{d/2}
 \right) .\notag 
\end{align}
It follows that the stationary (probability) density for the position of the particle is given by
\begin{align*}
h_\ell\left(z,u\right)
=C_1^{-1}\left(E^F +   m \omega^2 \|z^H \|^2/2\right)^{d/2-1} , 
\qquad z\in D_\ell.
\end{align*}
This proves \eqref{f22.4}.

(iii)
Consider the process $x(t)$ in the torus $D_T$ and
assume that the process is in the stationary regime. 
Note that the stationary distribution of $x(t)$ in $D_T$ is given by \eqref{f22.4} with $\ell=1$.

Consider a small time
interval $[0,\Delta t]$. The particle can collide with  
$\prt_c D_T$ in this time interval only if  its 
distance from $\prt_c D_T$ at time $t=0$  is less than $\eps:=v_* \Delta t$,
where $v_*$ is the maximum speed in $F$, given by \eqref{fe25.4}.   
The probability that the particle is within distance $\eps$ from $\prt_c
D_T$ at time 0 is equal to $a+ O(\eps^2)$, where 
\begin{align}\label{f23.2}
a &= h_\ell((\rho, 0,\dots, 0),u)(2\ell)^{d-2} 2\pi \rho\eps\\
&=
\frac{d m \omega^2 }{4(2\ell)^{d-2} \pi}
\cdot
 \frac{
\left(E^F +   m \omega^2 \|z^H \|^2/2\right)^{d/2-1}}
{
 \left( E^F + m \omega^2 \rho^2/2 \right)^{d/2}
 - \left( E^F  \right)^{d/2}}(2\ell)^{d-2} 2\pi \rho\eps \notag\\
&=
\frac{d m \omega^2\rho\eps}{2}
\cdot
\frac
{\left(E^F +   m \omega^2 \rho^2/2\right)^{d/2-1}}
{\left( E^F + m \omega^2 \rho^2/2 \right)^{d/2} 
 - \left( E^F  \right)^{d/2}}
.\notag
\end{align}
We will use the formulas $\eps=v_* \Delta t$, \eqref{fe25.4} and
\eqref{f23.2} in the following calculation. 
The probability of a reflection from $\prt_c D$ during the interval
$[0,\Delta t]$ is approximately equal to, with accuracy $O((\Delta t)^2)$, 
\begin{align}
&\frac{a}{\calA(d-1,\eps)} \int_0^\eps \int_r^\eps \calA(d-2,(\eps^2 -
  s^2)^{1/2}) 
\frac{ds}{(\eps^2-s^2)^{1/2}}\,dr \notag \\ 
&\quad= 
\frac{a}{\calA(d-1,1)} \int_0^1 \int_r^1\calA(d-2,(1-s^2)^{1/2})\frac{ds}{(1-s^2)^{1/2}}\,dr \notag \\ 
&\quad=
\frac{a\calA(d-2,1)}{\calA(d-1,1)} \int_0^1
  \int_r^1(1-s^2)^{(d-3)/2}ds\,dr \notag \\ 
&\quad=
\frac{a\calA(d-2,1)}{\calA(d-1,1)} \int_0^1 \int_0^s
  (1-s^2)^{(d-3)/2}dr\,ds \notag \\
&\quad=
\frac{a\calA(d-2,1)}{\calA(d-1,1)} \int_0^1 s(1-s^2)^{(d-3)/2}ds \notag \\
&\quad=
\frac{a\calA(d-2,1)}{(d-1)\calA(d-1,1)} \notag \\
&\quad =
\frac{d m \omega^2 \rho\eps}{2}
\cdot 
\frac{\left(E^F +   m \omega^2 \rho^2/2\right)^{d/2-1}}
{\left( E^F + m \omega^2 \rho^2/2 \right)^{d/2}
 - \left( E^F  \right)^{d/2}}
 \cdot
\frac{\calA(d-2,1) }{(d-1)\calA(d-1,1)} \notag  \\
&\quad =
\sqrt{\frac{m}{2}}\cdot \omega^2 \rho \cdot 
\frac{\left(E^F +   m \omega^2 \rho^2/2\right)^{(d-1)/2}}
{\left( E^F + m \omega^2 \rho^2/2 \right)^{d/2}
 - \left( E^F  \right)^{d/2}}
 \cdot
\frac{d\calA(d-2,1) }{(d-1)\calA(d-1,1)}\Delta t \label{fe25.2}
\end{align}

By the
ergodic theorem, $\E\tau_T$ is arbitrarily close to the reciprocal of
the  quantity in \eqref{fe25.2} (except for $\Delta t$), i.e., it is
approximately equal to 
\begin{align}\label{fe25.3}
\sqrt{\frac{2}{m}}\cdot \frac1 {\omega^{2} \rho }
\cdot
\frac
{\left( E^F + m \omega^2 \rho^2/2 \right)^{d/2}
 - \left( E^F  \right)^{d/2}}
{ \left(E^F +   m \omega^2 \rho^2/2\right)^{(d-1)/2}}
\cdot
\frac{(d-1)\calA(d-1,1) }{d\calA(d-2,1)}.
\end{align}
It is easy to see that $\E\tau_\infty = \E\tau_T$  so the last formula
proves \eqref{f22.7}.

(iv)
In this case, the analogue of \eqref{fe25.1} is
\begin{align*}
C_2&:=\int _{\calQ_{d-2}(\ell)}
\int _{\calB_2(0,\rho)\setminus \calB_2(0,\rho_0)}
 \left(E^F +   m \omega^2 \|z^H \|^2/2\right)^{d/2-1} dz^H
 dz^\perp \\
 & = 
 (2\ell)^{d-2} \int _{\rho_0}^\rho
 2\pi r
 \left(E^F +   m \omega^2 r^2/2\right)^{d/2-1} dr\\
 &  = 
 (2\ell)^{d-2} \int _{\rho_0^2}^{\rho^2}
 \pi 
 \left(E^F +   m \omega^2 s/2\right)^{d/2-1} ds\\
&  = 
 \frac{4(2\ell)^{d-2} \pi}{d m \omega^2 }
 \left( E^F + m \omega^2 \rho^2/2 \right)^{d/2} .
\end{align*}
It follows that the stationary (probability) density for the position of the particle is given by
\begin{align}\label{f23.1}
h_\ell\left(z,u\right)
=C_2^{-1}\left(E^F +   m \omega^2 \|z^H \|^2/2\right)^{d/2-1} , 
\qquad z\in D_\ell.
\end{align}
This proves \eqref{f22.5}.

(v)
In the present case, we modify  \eqref{f23.2} as follows,
\begin{align*}
a_1 &= h_\ell((\rho, 0,\dots, 0),u)(2\ell)^{d-2} 2\pi \rho\eps\\ 
&=
\frac
{d m \omega^2 \rho \eps\left(E^F +   m \omega^2 \rho^2/2\right)^{d/2-1}} 
{2\left( E^F + m \omega^2 \rho^2/2 \right)^{d/2}}
=
\frac
{d m \omega^2 \rho\eps }
{2\left( E^F + m \omega^2 \rho^2/2 \right)}
.\notag
\end{align*}
Hence, the approximate probability of a reflection during the interval $[0,\Delta t]$
is given by a formula analogous to \eqref{fe25.2}, 
\begin{align*}
&(m/2)^{1/2} \omega^2 \rho
\left(E^F +   m \omega^2 \rho^2/2\right)^{-1/2} 
 \cdot
\frac{\calB(d-1,1) }{\calB(d,1)} 
\Delta t.
\end{align*}
This implies that  \eqref{fe25.3} should be modified as follows,
\begin{align*}
(2/ m)^{1/2} \omega^{-2} \rho^{-1}
\left( E^F + m \omega^2 \rho^2/2 \right)^{1/2}
\frac{\calB(d,1) }{\calB(d-1,1)}.
\end{align*}
Thus \eqref{f22.8} is proved.
\end{proof}

\begin{example}
(i) Suppose that $d=2$. 
The formulas in Theorem \ref{f22.3} take especially simple form in this case.
If $E^F\geq 0$ then $h_\ell\left(z,u\right)$ is the uniform density in $ D_\ell\times S$.
If $E^F<0$ then $h_\ell\left(z,u\right)$ is the uniform density in $\{z\in
D_\ell: \|z^H\| > \rho_0\}\times S$. 

Recall \eqref{fe25.4} to see that if $E^F\geq 0$, then \eqref{f22.7}
can be written
\begin{align*}
 \E \tau_\infty &=
\sqrt{\frac{2}{m}}\cdot 
\frac{1}{\omega^{2} \rho }
 \cdot
\frac{\calB(d,1) }{\calB(d-1,1)}
\cdot
\frac
{ m \omega^2 \rho^2 }
{ 2\left(E^F +   m \omega^2 \rho^2/2\right)^{1/2}}\\
&=\sqrt{\frac{2}{m}}
\cdot 
\frac{\pi}{2}\cdot\frac{m\rho}{2}\cdot v_*^{-1} (m/2)^{-1/2} 
 =\frac{\pi \rho } {2 v_*} . 
\end{align*}
If $E^F< 0$, then \eqref{f22.8} becomes
\begin{align*}
 \E \tau_\infty &=
 \frac{\pi v_*}{2\omega^2\rho}.
\end{align*}

It is natural  that $\E \tau_\infty$ goes to 0 as $v_*$ becomes very large
(because the trajectory crosses the cylinder very fast) and the same is
true  when $v_*$ goes to 0 (because the particle stays very close to 
$\prt_c D$ during the whole flight). It is less obvious that $\E
\tau_\infty$ should be a monotone function of $v_*$ (other parameters being
fixed) in each regime $E^F\geq 0$ and $E^F< 0$.  

Curiously, if we fix $\rho$ and $v_*$ then $\E \tau_\infty$ does not depend
(explicitly) on the angular velocity $\omega$ in the case $E^F >0$ but it
does when $E^F<0$. In the last case,  $\E \tau_\infty\to 0$ when
$\omega \to \infty$ because the centrifugal force keeps the particle
close to $\prt _c D$. 

(ii) If $E^F=0$, formulas \eqref{f22.4}-\eqref{f22.7} agree with
\eqref{f22.5}-\eqref{f22.8}, as expected, and take the form 
\begin{align*}
h_\ell\left(z,u\right)
&=
\frac{d  }{2(2\ell)^{d-2} \pi \rho^d}
  \|z^H \|^{d-2}, \qquad z\in D_\ell, \\
\E \tau_\infty &=
\frac{\calB(d,1) }{\calB(d-1,1)}
\cdot \frac1 {\omega }.
\end{align*}

(iii) It is easily seen that for large $d$, 
\[
\E \tau_\infty \sim \sqrt{2\pi}
\cdot\frac{v_*}{\omega^2\rho}
\cdot d^{-1/2} 
\] 
for all $E^F\in \mathbb{R}$ and $\rho>\rho_0$.  
\end{example}

\subsection{Rotation rate}
We will prove that the asymptotic rotation rate for a pointlike particle in
rotating drum is equal to the angular speed of the drum, for any drum
shape, assuming the Knudsen reflection law. Our proof applies to any
random reflection law that arises in Propositions \ref{j7.1} and
\ref{m16.2}, provided that the state space consists of one communicating
class (the process is neighborhood irreducible). 

Assume that the drum is bounded but has an arbitrary shape, as in Section
\ref{sec:gen}.  Recall that the rotation axis is orthogonal to the
$(z_1,z_2)$-plane and that the drum rotates with angular velocity
$\omega>0$ in $H$. If $x^H(t)\ne (0,0)$ for  
$t\in[0, s)$ then we can uniquely represent $x^H(t)$ on this time interval
  using the complex notation as $x^H(t) = \|x^H(t)\| \exp(i \Theta(t))$,
  $t\in[0, s)$, with the convention that $\Theta(0)=0$ and $\Theta(t)$ is
    continuous.
Since the reflections are Lambertian, the probability that $x^H$ will hit
$(0,0)$ is zero and, therefore, $\Theta(t)$ is well defined for all $t$,
a.s. 

\begin{proposition}
The limit $\lim_{t\to\infty} \Theta(t)/t$ exists and is equal to $ \omega$.
\end{proposition}

\begin{proof}
We define $\Theta^F(t)$ in a way analogous to that for $\Theta(t)$.
If $x^{F,H}(t)\ne (0,0)$ (this is equivalent to $x^H(t)\ne (0,0)$)
for 
$t\in[0, s)$ then we  uniquely represent $x^{F,H}(t)$ on $[0,s)$ as
    $x^{F,H}(t) = \|x^{F,H}(t)\| \exp(i \Theta^F(t))$, with the convention
    that $\Theta^F(0)=0$ and $\Theta^F(t)$ is continuous.
Note that $\Theta^F(t) = \Theta(t) - \omega t$, so it will suffice to
prove that $\lim_{t\to\infty} \Theta^F(t)/t=0$. 

In view of the spherical symmetry of the velocities distribution in $F$,
stated in \eqref{f13.12}, the ergodic theorem implies that if the limit
$\lim_{t\to\infty} \Theta^F(t)/t$ exists then it must be 0. 

Note that $\Theta$ can change by at most $\pi$ during one flight. Since $D$
is bounded and $\omega$ is constant, it follows that for some $c_1<\infty$,
$\Theta^F$ can change by at most $c_1$ during one flight. Hence, the
ergodic theorem applies. 

It remains to show that the measure in \eqref{f13.10}-\eqref{f13.11}, 
properly normalized, represents the unique stationary probability  
distribution for the process in the rotating frame
$F$. 
The locations of reflection point on the boundary form a neighborhood irreducible process, by assumption.
 Two independent copies of the process will
eventually  find themselves (not necessarily at the same time) in positions from where they can hit some parts
of the boundary with densities whose ratio is bounded away from zero and
infinity. This can be used to show that one can construct  two copies of the process that will couple in a
finite time, a.s. Therefore, their distributions must converge to the
same limiting distribution. The limit must be equal to the distribution
given in  \eqref{f13.10}-\eqref{f13.11}. 
\end{proof}

\begin{proposition}\label{f27.1}
If $E^F<0$ then $\frac {d}{dt} \Theta(t) >0$.
\end{proposition}

\begin{proof}
It follows from \eqref{Fenergies} that, given $x^H(t)$, the speed of the
particle in $F$ is  
\begin{align*}
\sqrt{\frac2{m}}\left(E^{F,K} \right)^{1/2}
= \sqrt{\frac2{m}}\left(E^F +   m \omega^2 \|x^H(t) \|^2/2 \right)^{1/2}. 
\end{align*}
If this speed is less than $\omega\|x^H(t)\|$ then $\Theta(t)$ must be
increasing. The condition  
\begin{align*}
\sqrt{\frac2{m}}\left(E^F +   m \omega^2 \|x^H(t) \|^2/2 \right)^{1/2}
< \omega\|x^H(t)\| 
\end{align*}
is equivalent to $ E^F<0$. 
\end{proof}

\section{Rotating billiards table in gravitational field}\label{sec:grav}

Consider a rotating two-dimensional billiards table immersed in the
gravitational field with a constant acceleration. 
We will show that there is no universal bound for the energy of the
billiard particle, in an appropriate sense, in two cases: 
(i) if the billiards table is circular, rotates about its center,
and the reflections are Lambertian;   or (ii) the reflections are specular
and the billiards table is a  smooth, arbitrarily small, deformation of a
disc. We will state these claims in a precise manner as Corollaries
\ref{o16.1} and \ref{o16.3} at the end of this section. 

The main technical results of this section are concerned with a model
different from any of the two models mentioned above but closely related to
them. 
Consider a two-dimensional billiards table in the shape of the disc with
center $(0,0)$ and radius 1, rotating around its center with the angular
velocity of $\omega> 0$ radians per time unit in the counterclockwise
(positive) direction.  
Assume that there is a gravitational field with constant acceleration,
parallel to the disc, with the gravitational acceleration equal to $-g$ for
some $g>0$, in the vertical direction. 
If $v(t) = (v_x (t), v_y(t))$ denotes the velocity of the particle  then
\begin{align*}
\frac\prt{\prt t} v_x(t) = 0, 
\qquad
\frac\prt{\prt t} v_y(t) = -g,
\end{align*}
for all $t$ that are not reflection times.

The above determines the trajectory between reflection times, assuming that
the reflection times and the velocities just after reflections are
given. In the
following lemma we consider the motion without any reflections (more precisely, reflections are irrelevant for this lemma). 

\begin{lemma}\label{ju21.1}
Consider
any pair of distinct points $p_1$ and $ p_2$ on the unit circle and
gravitational acceleration $g>0$. There exists $w_0<\infty$ such that for
any  $w\geq w_0$ there is a unique initial velocity 
$v(0) $ such that all of the following conditions are satisfied,

(i) $\|v(0)\| = w$, 

(ii) if
the particle starts from $p_1$ with velocity $v(0) $ then its trajectory will pass through $p_2$, 

(iii) the trajectory defined in (ii) will stay inside the open unit disc until it reaches $p_2$.

\end{lemma}

\begin{proof}
The proof is based on totally elementary calculations so we will only sketch the main steps.

First, 
it is quite obvious that for any given $p_1,p_2$ and $g$,
there exists $w_0<\infty$ such that for any  $w\geq w_0$ there exists (at least one) initial velocity
$v(0) $ such that $\|v(0)\| = w$ and if
the particle starts from $p_1$ at time 0 with velocity $v(0) $ then its
trajectory will pass through $p_2$. 

Second, it is easy to check that if $w$ is strictly greater than $w_0$
defined in the previous paragraph then there are exactly two  initial
velocities 
$\wt v(0) $ and $\wh v(0) $ such that $\|\wt v(0)\|=\|\wh v(0)\| = w$ and if
the particle starts from $p_1$ at time 0 with velocity $\wt v(0) $ or $\wh
v(0) $ then its trajectory will pass through $p_2$. 
Extend these parabolic trajectories to negative times. For exactly one of
these initial velocities, the highest point on the trajectory is attained
between the times when the trajectory passes through $p_1$ and $p_2$.  

Recall that the disc radius is 1.
One can find $w_1<\infty$ so large that if $\|v(0)\| \geq w_1$ then there
exists a trajectory such that its   highest point  is at least 3 units
above $p_1$, and it is attained at a time  between the hitting times of
$p_1$ and $p_2$. This trajectory does not satisfy condition (iii) of the
lemma, and, therefore, there is at most one trajectory satisfying (iii). 

We will argue that the other trajectory satisfies (iii) provided that $w_0$
is large enough. The chord $C$ joining $p_1$ and $p_2$ forms the same,
non-zero angle with the unit circle at both ends.  When $\|v(0)\|$
increases then the slope of the trajectory between $p_1$ and $p_2$
converges uniformly to the slope of $C$. This implies (iii). 
\end{proof}

Now we will define the reflection rules in the main model in this section. 

\begin{definition}\label{j25.2}
Consider
any pair of distinct points $p_1=(x_1,y_1)$ and $ p_2=(x_2,y_2)$ on the unit circle in non-rotating coordinate system. 
Let $s_1=0$ and suppose that the particle starts from $p_1$ at time
$t=0=s_1$. If the initial velocity is such that the trajectory satisfies
conditions (ii) and (iii) of Lemma \ref{ju21.1} then we let $s_2$ be the
hitting time of $p_2$. 
 Recall the definition of the rotating frame of reference $ F$ from Section \ref{RD}.
We reflect the particle at $p_2$ in such a way
that 
(a) the energy is conserved in $F$ and, (b)
 the trajectory satisfies  conditions (ii) and (iii) of Lemma \ref{ju21.1},
 with roles of $p_1$ and $p_2$ interchanged. 
We let $s_3$ be the hitting time of $p_1$.

We proceed by induction. Suppose that $s_{2k+1}$ has been defined and the
particle is at $p_1$ at time $s_{2k+1}$.  
We reflect the particle at $p_1$ in such a way that
(a) the energy is conserved in $F$ and, (b)
 the trajectory satisfies  conditions (ii) and (iii) of Lemma \ref{ju21.1}.
We let $s_{2k+2}$ be the hitting time of $p_2$.

If $s_{2k}$ has been defined and the particle is at $p_2$ at time $s_{2k}$
then we reflect the particle at $p_2$ in such a way that
(a) the energy is conserved in $F$ and, (b)
 the trajectory satisfies  conditions (ii) and (iii) of Lemma \ref{ju21.1},
 with roles of $p_1$ and $p_2$ interchanged. 
We let $s_{2k+1}$ be the hitting time of $p_1$.

The sequence $s_1, s_2, \dots$ might be finite, if conditions (ii) and
(iii) of Lemma \ref{ju21.1} cannot be satisfied at some stage. 
\end{definition}

\begin{proposition}\label{j13.2}
For any $\omega, g$ and
any pair of distinct points $p_1=(x_1,y_1)$ and $ p_2=(x_2,y_2)$ on the
unit circle there exists $w_0 < \infty$ such that the following holds.  

(i) If $|x_1|= |x_2|$ and $\| v(s_1+)\| \geq w_0$ then the sequence $s_1,
s_2, \dots$ is infinite and $v(s_{k+2}+) = v(s_k+)$ for all $k\geq 1$. 

(ii) If  $x_2 > x_1 \ne -x_2$ and $ v_x(s_1+) \geq w_0$ then the sequence
$s_1, s_2, \dots$ is infinite and for all $k\geq 2$, 
\begin{align}\label{j18.1}
v_x(s_{k+2}-) = v_x(s_{k}-) +  \frac{(-1)^k}{v_x(s_{k}-)^2}
\cdot \frac{g \omega(x_2-x_1)^3 ( x_1+x_2)}
{(x_2-x_1)^2 +(y_2-y_1)^2} + O\left(\frac1{|v_x(s_{k}-)|^3}\right).  
\end{align}

\end{proposition}

\begin{proof}

Recall that the
 billiards table is the disc with center $(0,0)$ and radius 1, rotating
 around its center with the angular velocity of $\omega$ radians per time
 unit in the counterclockwise (positive) direction. 
Hence, the billiards table boundary point which happens to be at  $p_1$ is
moving with the velocity $\omega(-y_1, x_1)$. The law of conservation
of energy in the moving frame of reference $F$ requires that the vectors
$v(s_k-) - \omega(-y_1, x_1)$ and $v(s_k+) - \omega(-y_1, x_1)$
have the same norm, for odd $k$. In other words, 
\begin{equation}\label{j13.3}
\big(v_x(s_k-)+\omega y_1\big)^2 + \big(v_y(s_k-)-\omega x_1\big)^2 
= \big(v_x(s_k+)+\omega y_1\big)^2 + \big(v_y(s_k+)-\omega x_1\big)^2.
\end{equation}
The particle reflects at $p_2$ at times $s_k$ for even $k$. The analogous formula to \eqref{j13.3} is
\begin{equation}\label{j13.4}
\big(v_x(s_k-)+\omega y_2\big)^2 + \big(v_y(s_k-)-\omega x_2\big)^2 
= \big(v_x(s_k+)+\omega y_2\big)^2 + \big(v_y(s_k+)-\omega x_2\big)^2.
\end{equation}

\medskip

We argue separately in the cases $x_1=x_2$ and $x_1\neq x_2$.  
In each part we will derive a
formula, similar in spirit to \eqref{j18.1}, relating speeds at
consecutive reflection times. If the speeds never decrease below a certain
threshold then we can appeal to Lemma \eqref{ju21.1} and conclude that 
the sequence $s_1, s_2, \dots$ is infinite.

\medskip
(i) 
Suppose that $x_1=x_2$. 
In this case we have $y_2=-y_1$ and  
\begin{align}\label{o18.2}
v_x(t)=0,\qquad t\geq 0.
\end{align}
 For $k=2$, the equation \eqref{j13.4} reduces to
\begin{equation}\label{j13.5}
\big(v_y(s_2-)-\omega x_2\big)^2 =  \big(v_y(s_2+)-\omega x_2\big)^2.
\end{equation}
Let $w =v_y(s_2-)$ and let $\delta$ be defined by
$v_y(s_2+)=-v_y(s_2-)+\delta = - w +\delta$. We solve \eqref{j13.5} for
$\delta$ as follows, 
\begin{align*}
&w^2+\omega^2x_2^2-2\omega x_2w=\delta^2 + (\omega x_2+ w)^2-2\delta
  (\omega x_2+w),\\
&\delta^2 -2\delta(\omega x_2+w) +4\omega x_2w = 0,\\
&\delta = 2 \omega x_2 \quad \text{  or  } \quad \delta = 2 w.
\end{align*}
If $\delta = 2 w$ then 
\begin{align*}
v_y(s_2+)=-v_y(s_2-)+2w = - w +2w = w = v_y(s_2-).
\end{align*}
This is impossible for geometric reasons. Hence, $\delta  =2\omega x_2$ and, therefore,
\begin{align}\label{o18.3}
v_y(s_2+)=-v_y(s_2-)+\delta = -v_y(s_2-)+ 2\omega x_2.
\end{align}
A similar argument shows that $v_y(s_3+)= -v_y(s_3-)+ 2\omega x_1$. Thus
\begin{align*}
v_y(s_3+)&= -v_y(s_3-)+ 2\omega x_1 = -v_y(s_2+)+ 2\omega x_1\\
&= -(-v_y(s_2-)+ 2\omega x_2) + 2\omega x_1 = v_y(s_2-) = v_y(s_1+).
\end{align*}
It is easy to see that the same argument applies to any $k$ and yields 
\begin{align}\label{o18.4}
v(s_{k+2}+) = v(s_k+), \qquad k\geq 1.
\end{align}
An argument similar to that in \eqref{o18.3} shows that
\begin{align*}
| v_y(s_{k+1}+) + v_y(s_k+)| \leq 2 \omega  |x_1|, \qquad k\geq 1.
\end{align*}
This, \eqref{o18.2} and \eqref{o18.4} imply that for any $w_0$ there exists
$w_1$ such that if $\| v(s_1+)\| \geq w_1$ then $\| v(s_k-)\| \geq w_0$ for
all $k$. This bound and Lemma \ref{ju21.1} prove that the sequence $s_1,
s_2, \dots$ is infinite. 

\medskip

(ii)
If $x_1\neq x_2$ then,
\begin{align}
\sign(v_x(0+))&=\sign(x_2-x_1) , \label{o18.1} \\
v_x(0+) s_2 &=  x_2-x_1 , \notag \\
s_2 = (x_2-x_1)/ v_x(0+) &= (x_2-x_1) / v_x(s_2-), \notag \\
v_y(0+) s_2 - (g/2)s_2^2 &= y_2-y_1, \notag \\
v_y(s_1+)=v_y(0+) = (g/2)s_2 + (y_2-y_1)/s_2 &= \frac{y_2-y_1}{x_2-x_1}
v_x(s_2-) + \frac{g(x_2-x_1)}{2v_x(s_2-)}, \notag \\ 
v_y(s_1+)&= \frac{y_2-y_1}{x_2-x_1} v_x(s_1+) + \frac{g(x_2-x_1)}{2v_x(s_1+)} ,\label{j14.1}\\
v_y(s_2-) = v_y(0+) - g s_2 &=  \frac{y_2-y_1}{x_2-x_1} v_x(s_2-) - \frac{g(x_2-x_1)}{2v_x(s_2-)}.
\label{j14.2}
\end{align}
Exchanging the roles of $p_1$ and $p_2$, and shifting time from $s_1$ to $s_2$, we 
obtain a formula analogous to \eqref{j14.1}, 
\begin{align*}
v_y(s_2+) = \frac{y_1-y_2}{x_1-x_2} v_x(s_2+) + \frac{g(x_1-x_2)}{2v_x(s_2+)}
= \frac{y_2-y_1}{x_2-x_1} v_x(s_2+) + \frac{g(x_1-x_2)}{2v_x(s_2+)}.
\end{align*}
This, \eqref{j13.4} and \eqref{j14.2} imply that
\begin{align}\label{Energy2}
&\left(v_x(s_2-)+\omega y_2\right)^2 +\left(\frac{y_2-y_1}{x_2-x_1} v_x(s_2-) 
- \frac{g(x_2-x_1)}{2v_x(s_2-)}-\omega x_2\right)^2 \\
&\quad= \left(v_x(s_2+)+\omega y_2\right)^2 
+\left(\frac{y_2-y_1}{x_2-x_1} v_x(s_2+) + \frac{g(x_1-x_2)}{2v_x(s_2+)}-\omega x_2\right)^2. \notag
\end{align}

We make the following definitions to simplify the notation,
\begin{align}
w & = v_x(s_2-),\label{o17.3}\\
\delta & = v_x(s_2-) + v_x(s_2+),\label{o17.4}\\
\lambda&=\frac{y_2-y_1}{x_2-x_1},\label{ju16.2}\\
\alpha&=x_2-x_1.\label{ju16.3}
\end{align}
Then $v_x(s_2+) = -w + \delta$ and  \eqref{Energy2} can be written in this form,
\begin{align}\label{ju15.1}
&\left(w+\omega y_2\right)^2 +\left(\lambda w 
- \frac{g\alpha}{2w}-\omega x_2\right)^2 \\
&\quad= \left(-w+\delta+\omega y_2\right)^2 
+\left(\lambda(-w+\delta) - \frac{g\alpha}{2(-w+\delta)}-\omega x_2\right)^2. \notag
\end{align}
This equation has to be also satisfied if $w = v_x(s_{2k}-)$ and 
$\delta = v_x(s_{2k}-) + v_x(s_{2k}+)$ for any integer $k\geq 1$. By symmetry,
the following condition
\begin{align}\label{ju15.2}
&\left(w+\omega y_1\right)^2 +\left(\lambda w 
+ \frac{g\alpha}{2w}-\omega x_1\right)^2 \\
&\quad= \left(-w+\delta+\omega y_1\right)^2 
+\left(\lambda(-w+\delta) + \frac{g\alpha}{2(-w+\delta)}-\omega x_1\right)^2, \notag
\end{align}
has to be satisfied if  $w = v_x(s_{2k+1}-)$ and  $\delta = v_x(s_{2k+1}-)
+ v_x(s_{2k+1}+)$ for  integer $k\geq 1$. 

We can think of \eqref{ju15.1} and \eqref{ju15.2} as equations with unknown
$\delta$, all other quantities being treated as known constants. It is
obvious that the equations are satisfied by  $\delta = 2w$. We will call a
solution \emph{relevant} if $\delta \ne 2 w$. 

Direct computations show that \eqref{ju15.1} is equivalent to
\begin{align*}
&\frac{(\delta-2 w) \left(\alpha^2 \delta g^2+4 \alpha g \omega  w x_2
    (\delta-w)-4 w^2 (\delta-w)^2 \left(\delta \lambda^2+\delta-2 \lambda
    \omega  x_2+2 \omega  y_2\right)\right)}{4 w^2 (\delta-w)^2} \\ 
&=0.
\end{align*}
Note that, since we require that the consecutive reflections of the
particle occur at $p_1$ and $p_2$, and $x_1\ne x_2$, we cannot have  
$v_x(s_2+) = 0$. Hence, $v_x(s_2+) = -w + \delta\ne 0$ and, therefore
$\delta-w\ne 0$. It follows that we are not dividing by 0 in the last
formula. Thus, to find a relevant $\delta$, we need to solve 
\begin{align}\label{ju15.3}
\alpha^2 \delta g^2+4 \alpha g \omega  w x_2 (\delta-w)-4 w^2
(\delta-w)^2 \left(\delta \lambda^2+\delta-2 \lambda \omega  x_2+2
\omega  y_2\right)=0. 
\end{align}
We are interested in solutions for $|w|$ large, so set $u=w^{-1}$. 
Equation \eqref{ju15.3} becomes 
\begin{equation}\label{intermsofu}
F(u,\de):=4(1-u\delta)^2\Big((1+\lambda^2)\delta+2\om(y_2-\la x_2)\Big) 
+4\al g \om x_2u^2(1-u\de) - \al^2g^2u^4\de=0.
\end{equation}
When $u=0$ the unique solution is 
\begin{equation*}
\delta = \delta_0: = \frac{2 \omega  (\lambda x_2 - y_2)}{1+\lambda^2}.
\end{equation*}
Observe that 
$\pa_\de F(0,\de_0)=4(1+\la^2)\neq 0$.  Therefore the implicit function   
theorem implies the existence of a unique real-analytic function 
$f(u)$ defined for $u$ near $0$ such that $f(0)=\de_0$ and $F(u,f(u))=0$.     
The derivatives of $f$ at $u=0$ can be obtained by successively
differentiating the equation $F(u,f(u))=0$ and evaluating at $u=0$.  
Differentiating once and using the product rule on the first term gives
\[
4\big((1-uf)^2\big)'\cdot 0 +4(1+\la^2)f'(0)+ 0+0=0,
\]
so $f'(0)=0$.  Differentiating twice, applying the product
rule twice on the first term and recalling $f'(0)=0$ gives 
\[
4\big((1-uf)^2\big)''\cdot 0 +8\big((1-uf)^2\big)'\cdot 0 
+4(1+\la^2)f''(0)+8\al g \om x_2 +0=0.
\]
So 
\[
f''(0) = -\frac{2\al g\om x_2}{1+\la^2}.
\]

Set $\de(w)=f(w^{-1})$.  Then for $|w|$ sufficiently large, $\de(w)$ is
given by a convergent power series in $w^{-1}$ and $\de = \de(w)$ solves 
\eqref{ju15.1}.  If we set 
\[
\de_2=\frac12 f''(0)=-\frac{\al g\om x_2}{1+\la^2},
\]
then the expansion of $\de(w)$ in powers of $w^{-1}$ is 
\begin{equation}\label{deltaexpand}
\de(w)= \de_0 +\de_2 w^{-2} + O(|w|^{-3}).
\end{equation}
Recall the definition \eqref{ju16.2} of $\lambda$ to see that 
\begin{align}\label{ju16.5}
\delta_0 = \frac{2 \omega  (\lambda x_2 - y_2)}{1+\lambda^2}
= \frac{2 \omega  (\frac{y_2-y_1}{x_2-x_1} x_2 - y_2)}
{1+\left(\frac{y_2-y_1}{x_2-x_1}\right)^2}
= \frac{2 \omega  (x_1 y_2 - x_2 y_1)(x_2-x_1)}
{(x_2-x_1)^2 +(y_2-y_1)^2}.
\end{align}
Likewise we use \eqref{ju16.2} and \eqref{ju16.3} to write 
\begin{align}\label{ju16.6}
\delta_2 = -\frac{\alpha g \omega  x_2}{1+\lambda^2}
= -\frac{(x_2-x_1) g \omega  x_2}
{1+\left(\frac{y_2-y_1}{x_2-x_1}\right)^2}
= -\frac{(x_2-x_1)^3 g \omega  x_2}
{(x_2-x_1)^2 +(y_2-y_1)^2}.
\end{align}

We use \eqref{o17.3}-\eqref{o17.4} and the generalization of this notation
to even indices, together with \eqref{ju16.5}-\eqref{ju16.6}, 
to write \eqref{deltaexpand} in the following form, 
\begin{align}\label{o17.5}
v_x(s_{2k+1}-) + v_x(s_{2k}-) & = \frac{2 \omega  (x_1 y_2 - x_2 y_1)(x_2-x_1)}
{(x_2-x_1)^2 +(y_2-y_1)^2}\\
&\quad   -\frac{(x_2-x_1)^3 g \omega  x_2}
{(x_2-x_1)^2 +(y_2-y_1)^2} \frac1{v_x(s_{2k}-)^2} +
O(|v_x(s_{2k}-)|^{-3}).\notag 
\end{align}
Since \eqref{ju15.2} can be obtained from \eqref{ju15.1} by exchanging the
roles of $(x_1,y_1)$ and $(x_2, y_2)$, it follows from  \eqref{o17.5} that  
\begin{align}\label{o17.6}
v_x(s_{2k+2}-) + v_x(s_{2k+1}-)
& = \frac{2 \omega  (x_1 y_2 - x_2 y_1)(x_2-x_1)}
{(x_2-x_1)^2 +(y_2-y_1)^2}\\
&\quad   +\frac{(x_2-x_1)^3 g \omega  x_1}
{(x_2-x_1)^2 +(y_2-y_1)^2} \frac1{v_x(s_{2k+1}-)^2} +
O(|v_x(s_{2k+1}-)|^{-3}).\notag
\end{align}
Subtracting \eqref{o17.5} from \eqref{o17.6} yields,
\begin{align}\label{o17.9}
v_x(s_{2k+2}-) - v_x(s_{2k}-)
& = \frac{(x_2-x_1)^3 g \omega  }
{(x_2-x_1)^2 +(y_2-y_1)^2} 
\left(\frac{x_1}{v_x(s_{2k+1}-)^2}
+ \frac{x_2}{v_x(s_{2k}-)^2}\right)\\
&\qquad + O(|v_x(s_{2k+1}-)|^{-3})
+ O(|v_x(s_{2k}-)|^{-3}).\notag 
\end{align}
Note that the signs of $v_x(s_{2k+1}-) $ and $ v_x(s_{2k}-)$ are different
and the same observation applies to $v_x(s_{2k+2}-) $ and $
v_x(s_{2k+1}-)$. 
It follows from \eqref{o17.5} that there exists $w_2$ such that if $|
v_x(s_{2k}-)| \geq 2w_2$ then 
\begin{align}\label{o17.7}
\Big||v_x(s_{2k+1}-)| -| v_x(s_{2k}-)|\Big| 
\leq \frac{3 \omega |(x_1 y_2 - x_2 y_1)(x_2-x_1)|}
{(x_2-x_1)^2 +(y_2-y_1)^2}.
\end{align}
We increase $w_2$, if necessary, so that \eqref{o17.7} combined with \eqref{o17.6} yields $| v_x(s_{2k+1}-)| \geq w_2$ and
\begin{align}\label{o17.8}
\Big||v_x(s_{2k+2}-)| -| v_x(s_{2k+1}-)|\Big| 
\leq \frac{4 \omega |(x_1 y_2 - x_2 y_1)(x_2-x_1)|}
{(x_2-x_1)^2 +(y_2-y_1)^2}.
\end{align}
In view of \eqref{o17.7}-\eqref{o17.8}, formula \eqref{o17.9} implies that
there exists $w_3$ such that, assuming $| v_x(s_{2k}-)| \geq w_3$, we have 
\begin{align}\label{ju21.2}
v_x(s_{2k+2}-) = v_x(s_{2k}-) + \frac{g \omega (x_2-x_1)^3 ( x_1+x_2)}
{v_x(s_{2k}-)^2\left((x_2-x_1)^2 +(y_2-y_1)^2\right)} +
O\left(\frac1{|v_x(s_{2k}-)|^3}\right). 
\end{align}
This proves \eqref{j18.1} for even $k$ and for all $x_2\neq x_1$. The proof
for odd $k$ is analogous. 

It follows from
\eqref{j18.1}  that for some $w_4$, if $| v_x(s_{1}+)| \geq w_4$ then at
least one of the sequences $\{|v_x(s_{2k}-)|, k\geq 1\}$ and
$\{|v_x(s_{2k+1}-)|, k\geq 0\}$ is nondecreasing,  
depending on the sign of 
\begin{align*}
\frac{g \omega (x_2-x_1)^3 ( x_1+x_2)}
{(x_2-x_1)^2 +(y_2-y_1)^2}.
\end{align*}
This and \eqref{o17.7}-\eqref{o17.8} imply that for some $w_5$, if 
$|v_x(s_{1}+)| \geq w_5$ then  $|v_x(s_{k}-)|\geq w_4$ for $k\geq 2$. We
can choose arbitrarily large $w_4$, so we can apply Lemma \ref{ju21.1} at all
reflection times. Hence, if $x_2\neq x_1$ then the sequence 
$s_1, s_2, \dots$ is infinite.  

It remains to show that if $x_1=-x_2\neq 0$, then $v(s_{k+2}+) = v(s_k+)$
for all $k\geq 1$.  If we take $w=v_x(s_{2k}-)$ with $|w|$ sufficiently
large, then $\de=\de(v_x(s_{2k}-))$ solves \eqref{ju15.1} and 
\begin{align*}
v_x(s_{2k+1}-) = -v_x(s_{2k}-) +\de(v_x(s_{2k}-)).
\end{align*}
The two sides of \eqref{ju15.1} are symmetric, in the sense
that exchanging the roles of $w$ and $-w+\delta$ turns the left hand side
into the right hand side and vice versa.  Therefore $\de=\de(v_x(s_{2k}-))$
also solves \eqref{ju15.1} for 
$w=-v_x(s_{2k}-) +\de(v_x(s_{2k}-))=v_x(s_{2k+1}-)$.  
According to Lemma \ref{j17.3} proved below, it follows that
$\de=\de(v_x(s_{2k}-))$ solves \eqref{ju15.2} for 
$w=v_x(s_{2k+1}-)$.  The same implicit function theorem argument applied
above to \eqref{ju15.1} can be applied to \eqref{ju15.2}, implying the
existence of a unique solution near $\de_0$ for $|w|$ large.  Thus   
\begin{align*}
v_x(s_{2k+2}-) = -v_x(s_{2k+1}-) +\de(v_x(s_{2k}-)).
\end{align*}
So, for all $k$,
\begin{align*}
v_x(s_{2k+2}-) & =  -v_x(s_{2k+1}-) +\de(v_x(s_{2k}-))\\
& =  -(-v_x(s_{2k}-) +\de(v_x(s_{2k}-))) + \de(v_x(s_{2k}-))\\
& =v_x(s_{2k}-). 
\end{align*}
The proof that $v_x(s_{k+2}-) = v_x(s_{k}-)$ for $k$ odd is analogous.  
\end{proof}

\begin{lemma}\label{j17.3}

If $x_1=-x_2\neq 0$, then any solution $(w,\delta)$ to
\eqref{ju15.1} is a solution to \eqref{ju15.2} and vice versa. 
\end{lemma}

\begin{proof}
The difference of the left hand sides of \eqref{ju15.1} and \eqref{ju15.2} is equal to
\begin{align}\label{ju16.10}
&\left(w+\omega y_2\right)^2 +\left(\lambda w 
- \frac{g\alpha}{2w}-\omega x_2\right)^2 
-\left(w+\omega y_1\right)^2 
-\left(\lambda w 
+ \frac{g\alpha}{2w}-\omega x_1\right)^2\\
&= \left(\omega y_2\right)^2 - \left(\omega y_1\right)^2
+ 2 w\omega (y_2-y_1)
+ \left(\lambda w -\omega x_2\right)^2
-\left(\lambda w -\omega x_1\right)^2\notag\\
&\quad -\frac{g\alpha}{w} \left(\lambda w -\omega x_2 +\lambda w -\omega x_1\right)\notag\\
&= \left(\omega y_2\right)^2 + \left(\omega x_2\right)^2
- \left(\omega y_1\right)^2 -\left(\omega x_1\right)^2
+ 2 w\omega (y_2-y_1)
-2 \lambda w \omega (x_2-x_1)\notag\\
&\quad -\frac{g\alpha}{w} \left(2\lambda w -\omega (x_1 +x_2)\right).\notag
\end{align}
Since the points $p_1=(x_1,y_1)$ and $ p_2=(x_2,y_2)$ lie on the unit circle,
\begin{align}\label{ju16.11}
\left(\omega y_2\right)^2 + \left(\omega x_2\right)^2
= \left(\omega y_1\right)^2 +\left(\omega x_1\right)^2 = \omega^2.
\end{align}
Recall that $\lambda = \frac{y_2-y_1}{x_2-x_1}$ to see that
\begin{align*}
2 w\omega (y_2-y_1)
-2 \lambda w \omega (x_2-x_1)
= 2 w\omega (y_2-y_1)
-2 \frac{y_2-y_1}{x_2-x_1} w \omega (x_2-x_1)=0.
\end{align*}
This, \eqref{ju16.10}, \eqref{ju16.11} and the assumption that $x_1 +x_2=0$ imply that
the difference of the left hand sides of \eqref{ju15.1} and \eqref{ju15.2} is equal to
\begin{align}\label{ju16.12}
 -\frac{g\alpha}{w} \left(2\lambda w -\omega (x_1 +x_2)\right)
= -2 g\alpha \lambda
+\frac{g\alpha\omega (x_1 +x_2)}{w} =-2 g\alpha \lambda.
\end{align}
Note that the right hand side of \eqref{ju15.1} can be obtained from the
left hand side by replacing $w$ with $-w+\delta$, and the same remark
applies to \eqref{ju15.2}. Since the right hand side of \eqref{ju16.12}
does not depend on $w$, it follows that the difference of the right hand
sides of \eqref{ju15.1} and \eqref{ju15.2} is equal to 
$-2 g\alpha \lambda$. Hence, both differences are equal to each other. This
proves the lemma. 
\end{proof}

\begin{corollary}\label{j25.1}
Suppose that assumptions of Proposition \ref{j13.2} (ii) hold; in
particular, $x_2 > x_1\ne - x_2$. There exists $w_0>0$ such that if
$v_x(s_1+)>w_0$   then  
\begin{align}\label{o18.10}
\lim_{t\to\infty} \frac{\|v(t-)\|^2}{t}
= g \omega  | x_1+x_2|.
\end{align}
\end{corollary}

\begin{proof}
We will give the proof in the case $x_1+x_2 >0$. The other case follows by symmetry.

Let $w_0$ be as in the statement of Proposition \ref{j13.2} (ii).

Under the assumptions of the corollary, it follows from \eqref{j18.1} that if $k$ is even and
\begin{align}\label{o18.11}
c_1 := \frac{g \omega (x_2-x_1)^3 ( x_1+x_2)}
{(x_2-x_1)^2 +(y_2-y_1)^2} >0,
\end{align}
then, 
\begin{align*}
v_x(s_{k+2}-) = v_x(s_{k}-) +  \frac{c_1}{v_x(s_{k}-)^2} + O\left(\frac1{|v_x(s_{k}-)|^3}\right).
\end{align*}
It follows, by comparison with the solution to the differential equation
$v'= c_1/v^2$, that for arbitrarily small $\eps>0$, there exist $k_0$ such
that for even $k\geq k_0$, 
\begin{align}\label{j19.1}
(1-\eps) (3 c_1 k/2)^{1/3} \leq
|v_x(s_{k}-)| \leq (1+\eps) (3 c_1 k/2)^{1/3}.
\end{align}
The first terms on the right hand sides of \eqref{o17.5} and
\eqref{o17.6} do not depend on $v_x(s_{2k}-)$ and $v_x(s_{2k+1}-)$ so \eqref{j19.1} holds not only for
even $k$ but for odd $k$ as well (although $k_0$ might have to be
adjusted). 

Since $s_{k+1} - s_k =   |x_2 - x_1| / |v_x(s_k+)|$, \eqref{j19.1} implies
that for any $\eps>0$  and sufficiently large $k$, 
\begin{align*}
s_{k+1} - s_k 
&\leq  \frac{ x_2 - x_1}{(1-\eps) (3 c_1 k/2)^{1/3}}
=
\frac{ (x_2 - x_1)(2((x_2-x_1)^2 +(y_2-y_1)^2))^{1/3}}
{(1-\eps) (3 g \omega (x_2-x_1)^3 ( x_1+x_2) k)^{1/3}}\\
&\leq (1+2\eps)
\left(\frac{2((x_2-x_1)^2 +(y_2-y_1)^2)}
{3 g \omega  ( x_1+x_2) k}\right)^{1/3}.
\end{align*}
The corresponding lower bound is
\begin{align*}
s_{k+1} - s_k 
\geq (1-2\eps)
\left(\frac{2((x_2-x_1)^2 +(y_2-y_1)^2)}
{3 g \omega  ( x_1+x_2) k}\right)^{1/3}.
\end{align*}
These bounds imply that for arbitrarily small $\eps>0$ and large $k$,
\begin{align}\label{jy10.1}
&(1-\eps)
\left(\frac{(x_2-x_1)^2 +(y_2-y_1)^2}
{ g \omega  ( x_1+x_2) }\right)^{1/3} \left(\frac 3 2\right)^{2/3} k^{2/3}\\
&\qquad\leq
s_k \leq (1+\eps)
\left(\frac{(x_2-x_1)^2 +(y_2-y_1)^2}
{ g \omega  ( x_1+x_2) }\right)^{1/3} \left(\frac 3 2\right)^{2/3} k^{2/3}.
\notag
\end{align}
This, in turn, means that
 for arbitrarily small $\eps>0$ and large $k$,
\begin{align*}
&(1-\eps)
\frac 2 3
\left(\frac{(x_2-x_1)^2 +(y_2-y_1)^2}
{ g \omega  ( x_1+x_2) }\right)^{-1/2}  s_k^{3/2}\\
&\qquad\leq
k \leq (1+\eps)\frac 2 3
\left(\frac{(x_2-x_1)^2 +(y_2-y_1)^2}
{ g \omega  ( x_1+x_2) }\right)^{-1/2}  s_k^{3/2}.
\end{align*}
We combine this estimate with \eqref{o18.11} and \eqref{j19.1} to conclude that
\begin{align*}
|v_x(s_{k}-)| &\leq (1+\eps) 
\left(\frac32 \frac{g \omega (x_2-x_1)^3 ( x_1+x_2)}
{(x_2-x_1)^2 +(y_2-y_1)^2} 
 \frac 2 3
\left(\frac{(x_2-x_1)^2 +(y_2-y_1)^2}
{ g \omega  ( x_1+x_2) }\right)^{-1/2}  s_k^{3/2}\right)^{1/3}\\
&= (1+\eps) 
\left( \frac{g \omega (x_2-x_1)^3 ( x_1+x_2)}
{(x_2-x_1)^2 +(y_2-y_1)^2} 
\left(\frac{(x_2-x_1)^2 +(y_2-y_1)^2}
{ g \omega  ( x_1+x_2) }\right)^{-1/2}  s_k^{3/2}\right)^{1/3}\\
&= (1+\eps)  (x_2-x_1)
\left( \frac{g \omega  ( x_1+x_2)}
{(x_2-x_1)^2 +(y_2-y_1)^2} \right)^{1/2}s_k^{1/2}.
\end{align*}
This, an analogous lower bound and \eqref{jy10.1} show that
\begin{align*}
\lim_{t\to\infty} \frac{|v_x(t-)|}{\sqrt{t}}
=  (x_2-x_1)
\left( \frac{g \omega  ( x_1+x_2)}
{(x_2-x_1)^2 +(y_2-y_1)^2} \right)^{1/2}.
\end{align*}
It follows from \eqref{j14.2} that
\begin{align*}
\lim_{t\to\infty} \frac{|v_y(t-)|}{\sqrt{t}}
=  (y_2-y_1)
\left( \frac{g \omega  ( x_1+x_2)}
{(x_2-x_1)^2 +(y_2-y_1)^2} \right)^{1/2}.
\end{align*}
Hence, we have the following asymptotic behavior for the energy $\|v(t-)\|^2$,
\begin{align*}
\lim_{t\to\infty} \frac{\|v(t-)\|^2}{t}
=
\lim_{t\to\infty} \frac{v_x(t-)^2+v_y(t-)^2}{t}
= g \omega  ( x_1+x_2).
\end{align*}
\end{proof}

The following two corollaries are concerned with models, mentioned at the
beginning of the section, that are more realistic than the model considered
in Proposition \ref{j13.2} and Corollary \ref{j25.1}. 

\begin{corollary}\label{o16.1}
Suppose that the billiards table is circular and reflections are
Lambertian, with the energy of the particle conserved in $F$ for every
reflection time $t$. 
For every rotation speed $\omega$ and acceleration $g>0$
there exist $w_0<\infty$ and a point $p_1$ on the unit circle such that for
every $w_1\in(w_0,\infty)$ there exists $\eps>0$ such that the particle
starting at $p_1$ with energy $w_0$ will attain energy $w_1$  with
probability greater than $\eps$. 
\end{corollary}

\begin{proof}

\emph{Step 1}.
According to Corollary \ref{j25.1}, there exist points $p_1$ and $p_2$ on
the unit circle and $w_0<\infty$ such that if the particle starts from
$p_1$ with the initial speed in the $x$-direction greater than $w_0$, and
reflecting according to the model described in Definition \ref{j25.2} then
for every $w_1>w_0$, the energy of the particle will be greater than $2
w_1$  at some  time $t_*$.  
We will call this trajectory $\calT$.
Suppose that $\calT$ undergoes exactly $m$ reflections before  $t_*$.

\medskip 
\noindent
\emph{Step 2}.
It is easy to see that
each part of the trajectory, between two consecutive reflections, is a
continuous function of the initial conditions in the following
sense. Suppose that the trajectory starts from point $p_3$ on the unit
circle at time $t=0$ with the initial velocity $v(0+)$ and it hits the unit
circle at time $t_1>0$, at a point $p_4$, with the velocity
$v(t_1-)$. Moreover, assume that the velocity $v(t_1-)$ is such that if the
trajectory continues after time $t_1$ then it will immediately cross to the
exterior of the unit disc. 
Then for every $\eps_1>0$ there is $\delta_1>0$ such that if a trajectory
starts from a  point $p_5$ on the unit circle at time $t_2$ (positive or
negative),  with the initial velocity $\wt v(t_2+)$, satisfying $\|p_3 -
p_5\| \leq \delta_1$, $|t_2| \leq \delta_1$, and $\|\wt v(t_2+) - v(0+)\|
\leq \delta_1$ then the new trajectory will hit the unit circle at a time
$t_3>0$, at a point $p_6$, with  velocity $v(t_3-)$, satisfying $\|p_6 -
p_4\| \leq \eps_1$, $|t_3-t_1| \leq \eps_1$, and $\|\wt v(t_1-) - v(t_3-)\|
\leq \eps_1$. 

\medskip 
\noindent
\emph{Step 3}.
Let $\theta_1$ denote a small (positive or negative) angle.
Let $\calT_1=\calT_1(\theta_1)$ be the trajectory  such that the velocity
vector of the particle just after $s_1 = 0$ forms the angle $\theta_1$ with
the velocity of the particle represented by $\calT$. The energy of the
particle in $F$ is assumed to be the same in both cases, $\calT$ and
$\calT_1$. The evolution of $\calT_1$ is governed by Definition
\ref{j25.2}. It follows from Step 2 and an induction argument that if
$|\theta_1|$ is sufficiently small then the energy of the particle  will be
greater than $(1+2^{-1}) w_1$ at some  time $u_1$, and, moreover, $\calT_1$
will have $m$ reflections before $u_1$. 

We proceed by induction. 
Suppose that $\calT_k=\calT_k(\theta_1, \dots, \theta_k)$ has been defined,
with the definition depending on parameters $\theta_1, \theta_2, \dots,
\theta_k$, and if $|\theta_j|$ is sufficiently small for $j=1,\dots,k$,
then the energy of the particle following $\calT_k$ is greater than
$(1+2^{-k}) w_1$ at some  time $u_k$. 
Let $s_{k+1}$ be defined as in  Definition \ref{j25.2}, relative to $\calT_k$.
Suppose that $\theta_{k+1}$ is a small, positive or negative, angle.
Let $\calT_{k+1}$ be a trajectory 
equal to $\calT_k$ until time $s_{k+1}$ and such that
the velocity vector in $\calT_{k+1}$ just after $s_{k+1} $ forms the angle
$\theta_{k+1}$ with the velocity of the particle represented by
$\calT_k$. The energy of the particle in $F$ is assumed to be the same in
both cases, $\calT_k$ and $\calT_{k+1}$. The evolution of $\calT_{k+1}$ is
governed by Definition \ref{j25.2} after time $s_{k+1}$. It follows from
Step 2 and an induction argument that  
if $|\theta_j|$ is sufficiently small for $j=1,\dots,k+1$, then the energy
of the particle following $\calT_{k+1}$ is greater than $(1+2^{-k-1}) w_1$
at some  time $u_{k+1}$, and, moreover, $\calT_{k+1}$ has $m$ reflections
before $u_{k+1}$. 

Let $\eta>0$ be so small that  if $|\theta_j|< \eta$  for $j=1,\dots,m$,
then the energy of the particle following $\calT_m=\calT_m(\theta_1, \dots,
\theta_m)$  is greater than $(1+2^{-m}) w_1$ at some  time $u_{m}$ and
$\calT_{m}$ has $m$ reflections before $u_{m}$. 
We now consider the model with Lambertian reflections, with the particle
starting from $p_1$. It is easy to see that there is a strictly positive
probability that the random trajectory generated in this way will be equal
to $\calT_m=\calT_m(\theta_1, \dots, \theta_m)$ for some $\theta_1, \dots,
\theta_m$ satisfying $|\theta_j|< \eta$  for $j=1,\dots,m$. Hence, there is
strictly positive probability that the trajectory with Lambertian
reflections will attain energy greater than $w_1$. 
\end{proof}

\begin{corollary}\label{o16.3}
For every rotation speed $\omega$ and acceleration $g>0$
there exists $w_0<\infty$ such that for every $w_1<\infty$ and $\eps>0$ 
there exist a billiard table whose $C^\infty$-smooth boundary lies inside
the annulus $\{x\in \R^2: 1< |x| < 1+\eps\}$, and  a trajectory with
specular reflections for the particle starting with energy $w_0$, such that
the energy of the particle will exceed $w_1$ at some time. 
\end{corollary}

\begin{proof}
Recall trajectory $\calT$ and the associated definitions ($p_1, p_2$, etc.) from
Step 1 of the proof of Corollary \ref{o16.1}. Let $v(t)$ denote the
velocity of $\calT$ at time $t\geq 0$. 
Let $t_1, t_2,  \dots$ be the times of reflections of $\calT$.
It follows from the continuity property of reflected billiard trajectories
discussed in Step 2 of the proof of Corollary \ref{o16.1} that for every
$n$ and $\eps>0$ there exist 
 sequences of distinct points $q_1, q_2,  \dots, q_n$ and times $s_1, s_2, \dots, s_n$, such that
\begin{align*}
\|q_{2k+1} - p_1 \| &< \eps, \qquad 1\leq 2k+1  \leq n,\\
\|q_{2k} - p_2 \| &< \eps, \qquad 1\leq 2k \leq n,\\
|s_k - t_k| &< \eps,\qquad 1\leq k\leq n,
\end{align*}
and a trajectory $\wt\calT$ with velocity $\wt v(t)$ such that it reflects
at times $s_k$ at points $q_k$, and 
\begin{align}\label{o16.2}
\|\wt v(s_k+) - v(t_k+) \| < \eps, \qquad 1\leq k \leq n.
\end{align}
Moreover, energy in the rotating frame of reference $F$ is conserved for the reflections of $\wt \calT$.

One can realize such collisions physically by perturbing the unit circle
near points $q_k$ locally (i.e., so that the perturbations around distinct
points $q_k$ do not overlap), in a $C^\infty$ way, so that the specular
reflection sends the reflecting billiard trajectory in the moving domain
from $q_k$ to $q_{k+1}$ at time $s_k$, for all $k$.  
Recall from
Step 1 of the proof of Corollary \ref{o16.1} that the energy of $\calT$ is
greater than $2 w_1$  at   time $t_*$ and $t_{m+1} \geq t_*$.  
If $n=m+1$ and $\eps>0$ is sufficiently small then it follows from
\eqref{o16.2} that the energy of $\wt\calT$ is greater than $ w_1$  at
time $t_*$.  
\end{proof}

\section{Microcanonical ensemble formula }\label{ME}
In this section we discuss the measure induced on a hypersurface by a 
smooth background measure and a defining function for the  
hypersurface.  In the setting of classical mechanics this provides an
invariant measure on an energy level surface (microcanonical ensemble),
which we make explicit for 
motion under the influence of a potential.  Undoubtedly this is 
well-known, but we have not found a suitable reference.  We describe two
specializations: first to a system of particles in a 
gravitational field, and second to the system discussed in this 
paper:  free particles viewed by an observer rotating at constant angular  
velocity.  In this section we use the methods and language of geometric
mechanics.  Background references for the discussion below are \cite{Lee} for
smooth manifolds and \cite{AM,Arn,MR} for geometric mechanics.      

Let $\cM$ be a smooth manifold, let $d\mu$ be a smooth measure on $\cM$
(i.e. a non-negative smooth density)
and let $H\in C^\infty(\cM,\R)$.  Let $\cM^r=\{dH\neq 0\}$ denote the
set of regular points of $H$.  For $E\in \R$, let $\cS_E=\{H=E\}$ be the 
level set of $H$ and let $\cS_E^r=\cS_E\cap\cM^r$ be the subset of regular 
points, a smoothly embedded hypersurface in $\cM^r$.  For each    
$E$, the pair $(d\mu,H)$ induces a measure on $\cS_E^r$ as follows.
If $\psi\in C_c(\cM^r)$ is a continuous function with compact support in  
$\cM^r$, then $\int_{H< t}\psi\, d\mu$ is a $C^1$ function
of $t$ which is non-decreasing if $\psi$ is non-negative.  For each $E$, the
assignment     
\[
\psi\mapsto \frac{d}{dt}\left(\int_{H< t}\psi\, d\mu\right)\Big|_{t=E}   
\]
defines a positive linear functional on $C_c(\cM^r)$; hence a Radon measure
on $\cM^r$ supported on $\cS_E^r$ which we will denote $d\Si_E$.
Another notation sometimes used for $d\Si_E$ is $\delta(H-E)d\mu$, where
$\delta$ denotes the Dirac delta function.  We can view $d\Si_E$ 
also as a measure on $\cS_E^r$.  No confusion should arise by using the
same notation $d\Si_E$ for both interpretations.  

It is clear from the definition that the construction of the measure
$d\Si_E$ is invariant under diffeomorphisms.  Namely, if $\ph:\cM\to\cM$ is a 
diffeomorphism, then for each $E$ the measure associated to 
$(\ph^*(d\mu), \ph^*H)$ is $\ph^*(d\Si_E)$.  Consequently if $d\mu$ and $H$
are invariant under $\ph$, then so is $d\Si_E$.   

The fundamental theorem of calculus implies that 
\begin{equation}\label{intdmu}
\int_{\cM^r}\psi\, d\mu 
= \int_{-\infty}^\infty\left(\int_{\cS_E^r}\psi\,d\Si_E\right)\,dE,\qquad
\psi\in C_c(\cM^r).
\end{equation}
Suppose $G$ is a Riemannian metric on $\cM$ and take $d\mu$ to be the  
Riemannian volume measure.  The coarea formula (see \cite{Chav}, Corollary
I.3.1) implies that the same equation \eqref{intdmu} holds, but with
$d\Si_E$ replaced 
by $d\cH/|\na H|$, where $d\cH$ denotes the surface measure (Hausdorff
measure) induced on $\cS_E^r$ by $G$, and $\na$ and $|\cdot|$  
are the gradient and norm relative to $G$.  Consequently  
$d\Si_E = d\cH/|\na H|$.  In particular, $d\cH/|\na H|$ is
independent of the choice of metric $G$ with volume form $dv_G=d\mu$.   

There is an equivalent realization of $d\Si_E$ in terms of  
differential forms (see \cite{AM}, Theorem 3.4.12).  Set $\dim \cM=D$.  If  
$\mu$ is a $D$-form on $\cM$, then it is easily seen that there is a 
unique $(D-1)$-form $\Si_E$ on $\cS_E^r$ with the property that if   
$\sbar$ is any $(D-1)$-form  
in a neighborhood of $\cS_E^r$ satisfying $\sbar|_{T\cS_E^r}=\Si_E$,
then $\mu = dH\wedge \sbar$ on $\cS_E^r$.  
If $d\mu=|\mu|$ is the measure determined by $\mu$, then $d\Si_E=|\Si_E|$. 
(Here $d\mu$ and $d\Si_E$ denote the measures discussed above, not the 
exterior derivatives of the differential forms.)   

Next let $(\cM^{2N},\Om)$ be a symplectic manifold with corresponding
volume form $\Om^N$ and volume measure $d\mu:=|\Om^N|$.  If  
$H\in C^\infty(\cM,\R)$, the associated Hamiltonian vector field $X_H$ is 
defined by $X_H\into \Om =-dH$.    
Since $H$ is constant along the flow $\ph_t$ of $X_H$, 
$\ph_t$ determines a flow $\ph_t|_{\cS^r_E}$ on $\cS^r_E$.    
\begin{proposition}
The measure $d\Si_E$ on $\cS_E^r$ determined by $d\mu$ and $H$ is     
invariant under $\ph_t|_{\cS^r_E}$.   
\end{proposition}
\begin{proof}
$\Om^N$ is $\ph_t$-invariant by Liouville's Theorem, and $dH$ is  
$\ph_t$-invariant since $H$ is.  It follows that $d\Si_E$ is invariant as 
well.  
\end{proof}

The invariant measure $d\Si_E$ can be written concretely in the   
setting of motion under the influence of a potential.  
Let $(M^N,g)$ be a Riemannian manifold and    
$V\in C^\infty(M)$.  The cotangent bundle $\cM:=T^*M$ has a canonical
symplectic structure given by $\Om = d\theta$, where $\theta=p_idq^i$ is
the tautological 1-form.  Here $q^i$ are local coordinates on $M$ and $p_i$
the corresponding dual coordinates on the fibers of $T^*M$.  In these
coordinates, $|\Om^N|$ has the form 
$|\Om^N|=c_Nd\lambda=c_Ndq^1\cdots dq^Ndp_1\cdots dp_N$, where $c_N>0$ is a  
constant depending only on $N$ and $d\lambda$ denotes $2N$-dimensional
Lebesgue measure.  

Consider a  Hamiltonian of the form   
\begin{equation}\label{hampot}
H(q,p)=\tfrac12 |p|^2_g +V(q)=\tfrac12 g^{ij}(q)p_ip_j +V(q).
\end{equation}
For $E>\inf V$, set $M_E=\{q:V(q)<E\}\subset M$ and 
\[
\cS_E^0=\{(q,p):H(q,p)=E\text{ and } p\neq 0\}\subset \cS_E^r,
\]
with projection $\pi:\cS_E^0\to M_E$.  For fixed $q\in M_E$, the fiber  
$\pi^{-1}(\{q\})$ is the sphere $\{p\in T^*_qM:|p|^2=2(E-V(q))\}$ of
radius $r(q)=\sqrt{2(E-V(q))}$.  The 
map $\Phi:  \cS_E^0\to S^*M_E$ given by $\Phi(q,p)=(q,p/|p|)$ is a  
diffeomorphism, where $S^*M_E=\{(q,p):q\in M_E, |p|=1\}$ is the unit
cosphere bundle over $M_E$.  We denote by $d\si_1(p) dv(q)$ the 
canonical measure on 
$S^*M_E$ determined by the volume measure $dv(q)$ of $g$ on $M_E$ and the 
usual surface measure $d\si_1(p)$ on each fiber $S^*_qM_E$ arising from its  
realization as the unit sphere in the Euclidean 
space $(T^*_qM, g)$.  (By abuse of notation, here we denote by $g$ also the 
inner product induced on $T_q^*M$.)        

\begin{proposition}\label{mainprop}
When restricted to $\cS_E^0$, the measure $d\Si_E$ 
determined by $|\Omega^N|$ and $H$ is given by  
\[
d\Si_E= c_N \big(2(E-V(q))\big)^{\frac{N}{2}-1}\Phi^*(d\si_1(p) dv(q)).      
\] 
In particular, this measure is the restriction to $\cS_E^0$ of a smooth
measure on $\cS_E^r$ which is invariant under the flow
$\varphi_t|_{\cS_E^r}$.   
\end{proposition}

\begin{proof}
We work locally over the domain of a coordinate chart in $M$.  Let $q^i$,
$1\leq i\leq N$ be local coordinates in $M_E$ and $p_i$ the induced linear  
coordinates on the fibers of $T^*M_E$.  Then 
$d\mu =c_N d\lambda$ is a constant multiple of Lebesgue measure.  
The metric 
$G=g_{ij}(q)dq^idq^j+g^{ij}(q)dp_idp_j$ has 
volume measure $d\lambda$.  According to the discussion above,  
$d\Si_E= c_N d\cH/|\na H|$, where $d\cH$ is surface measure on $\cS_E^r$ 
with respect to the metric $G$.  On $\cS_E^0$ this can be expressed as    
\[
d\cH=\frac{d\si_{r(q)}(p)dv(q)}{\cos\theta}.
\]
Here $d\si_{r(q)}(p)$ denotes surface measure on the sphere  
$\{p\in T^*_qM:|p|=r(q)\}$ with fixed $q$ and $\theta$ is the  
angle with respect to $G$ between the normals $\na H$ and $\na_p |p|^2$,  
where $\na_p |p|^2$ denotes the gradient in the $p$ variables with $q$ 
held fixed.  Now  
\[
\cos\theta 
= \frac{\langle \na H, \na_p |p|^2\rangle}{|\na H|\cdot |\na_p|p|^2|} 
= \frac{\langle \na_p |p|^2, \na_p |p|^2\rangle}{2|\na H|\cdot
  |\na_p|p|^2|}   = \frac{|\na_p|p|^2|}{2|\na H|} = \frac{|p|}{|\na H|}.   
\]
So 
\[
\begin{split}
c_N^{-1}d\Si_E&= d\cH/{|\na H|}=|p|^{-1}d\si_{r(q)}(p)dv(q)
=|p|^{N-2}\frac{d\si_{r(q)}(p)}{|p|^{N-1}}dv(q)\\
&=\big(2(E-V(q))\big)^{\frac{N}{2}-1}\Phi^*(d\si_1(p) dv(q)).       
\end{split}
\]
\end{proof}

We describe two examples in the next two sections. The first of these is
simpler and needed for a different project, presented in \cite{KBJM}. The second example is used in
this article.   

\subsection{Gravitational field}
First consider $n$ noninteracting point particles of masses  
$m_k>0$, $1\leq k\leq n$, moving in $\R^d$ under the influence of a
gravitational field imparting a constant acceleration {\tt g}.  Write
the position of the $k$-th particle as $x_k=(z_k,w_k)$ with $z_k\in \R$
$w_k\in \R^{d-1}$, where gravity acts in the downward $z_k$-direction.
Denote the velocity of the $k$-th particle by $v_k\in \R^d$ and its
momentum by $p_k=m_kv_k$.   
Set $\mathbf{x}= \big(x_1,x_2,\ldots,x_n\big)$ 
and $\mathbf{v}=(v_1,v_2,\ldots,v_n)$.  In the context of the above
discussion, take $M=\{\mathbf{x}\in\R^{nd}\}$ so that $N=nd$,
$q=\mathbf{x}$ and  $p=(p_1,\ldots,p_n)$.   
The metric $g$ is given by   
\begin{equation}\label{metric}
g(\mathbf{v},\mathbf{v})=\sum_{k=1}^nm_k\|v_k\|^2=\sum_{k=1}^n\frac{1}{m_k}\|p_k\|^2=:|p|^2_g, 
\end{equation}
where $\|\cdot\|$ denotes the Euclidean norm, and the potential $V$ is
given by 
\[
V=\text{{\tt g}}\sum_{k=1}^nm_k z_k.  
\]
Since $dV$ is nowhere vanishing, 
\[
\cS_E^r=\cS_E=\{(q,p):\tfrac12 |p|_g^2 +V(q)=E\}\quad\text{ and } \quad 
\cS_E^0=\{(q,p)\in \cS_E:p\neq 0\}.  
\] 
The unit cosphere bundle is $S^*M_E=\{(q,p): q\in M_E,\,\, |p|_g=1\}$.   
Let $S^{N-1}=\{\pb\in \R^N:\|\pb\|=1$\} denote the Euclidean usual unit
sphere in $\R^N$ and let $d\sigma_1(\pb)$ denote its usual measure.      
Define $\Psi:S^*M_E\to M_E\times S^{N-1}$ by
\begin{equation}\label{psi}
S^*M_E\ni (q,p)\stackrel{\Psi}{\mapsto} (q,\pb)\in M_E\times S^{N-1} 
\end{equation}
where 
\begin{equation}\label{pbar}
\pb=\left(\frac{p_1}{\sqrt{m_1}},\ldots,\frac{p_N}{\sqrt{m_N}}\right).
\end{equation}
If we set $\Phibar=\Psi\circ\Phi$, then Proposition~\ref{mainprop} implies  
\begin{proposition}\label{exampleprop}
The measure 
\begin{equation}\label{Phimeasure} 
\big(2(E-V(q))\big)^{\frac{N}{2}-1}\Phibar^*(d\sigma_1(\pb) dv(q))  
\end{equation}
is the 
restriction to $\cS_E^0$ of a smooth measure on $\cS_E$ which is invariant
under the flow $\varphi_t|_{\cS_E}$.       
\end{proposition}

\subsection{Rotating observer}\label{fe16.1}
This example is concerned with $n$ noninteracting free point particles 
of masses $m_k>0$, $1\leq k\leq n$, moving in $\R^d$, $d\geq 2$, but viewed
by an observer rotating with constant angular velocity $0<\om\in \R$.    
A discussion of motion observed by a rotating observer in the  
more general setting of time-dependent angular velocity vector and external 
force field can be found in \S 8.6 of \cite{MR}.   
First consider the case $n=1$ of a single particle.  Write its position as  
$x=(y,z,w)$ with $y$, $z\in \R$, $w\in \R^{d-2}$, and write
$x^H=(y,z,0)$ for its horizontal projection.  Set   
\[
L=\begin{pmatrix}
0&-1&0\\
1&0&0\\
0&0&0
\end{pmatrix}, \qquad 
A_t=\exp(t\om L)=
\begin{pmatrix}
\cos(\om t)&-\sin(\om t)&0\\
\sin (\om t)&\cos(\om t)&0\\
0&0&I
\end{pmatrix},
\]
where the $1\times 1\times (d-2)$ block decomposition corresponds to the 
decomposition of $x$.  The position as viewed by the rotating 
observer is $x^F=(y^F,z^F,w^F)$, where $x=A_tx^F$.  Therefore    
\[
\begin{split}
\dot{x}&=\dot{A_t}x^F + A_t\dot{x}^F\\ 
&=\dot{A_t}A_t^{-1}x + A_t\dot{x}^F\\ 
&=\om Lx + A_t\dot{x}^F
\end{split}
\]
and
\begin{equation}\label{transform}
\begin{split}
\ddot{x}&=\om L\dot{x} +\dot{A_t}\dot{x}^F 
+A_t\ddot{x}^F\\
&=\om L\big(\om Lx + A_t\dot{x}^F\big) +\om LA_t\dot{x}^F 
+A_t\ddot{x}^F\\
&=\om^2 L^2x + 2\om LA_t\dot{x}^F+A_t\ddot{x}^F.   
\end{split}
\end{equation}
Now $\ddot{x}=0$ since the particle is free.  So multiplying
\eqref{transform} by $A_t^{-1}$ shows that the equation of motion as
viewed by the rotating observer is
\begin{equation}\label{eom}
\ddot{x}^F=\om^2x^{F,H} -2\om L\dot{x}^F.   
\end{equation}
The first term on the right-hand side is the observed centrifugal force and
the second term the Coriolis force.   

Equation \eqref{eom} is clearly equivalent to the first order system
\begin{equation}\label{firstorder}
\dot{q}=m^{-1}p,\qquad \dot{p}=m\om^2q^H-2\om Lp
\end{equation}
via $q=x^F$.  Now \eqref{firstorder} is Hamiltonian, but with respect 
to the symplectic form 
\[
\Ot = \Om +2m\om \, dy^F\wedge dz^F
\]
rather than the canonical symplectic form $\Om$ on $T^*\R^d$.  
In fact, it is easily verified that the vector field
\begin{equation}\label{X}
X=m^{-1}p\cdot \pa_q+ (m\om^2q^H-2\om Lp)\cdot \pa_p
\end{equation}
satisfies $X\into \Ot = -dH$, where the Hamiltonian function is      
\begin{equation*}
H(q,p)=\frac{1}{2m} \|p\|^2 -\frac{m}{2}\om^2 \|q^H\|^2 
=\frac{1}{2m} \|p\|^2 -\frac{m}{2}\om^2 \left[(y^F)^2+(z^F)^2\right].  
\end{equation*}
Thus $X=X_H$ is the Hamiltonian vector field associated to $H$ by the 
symplectic form $\Ot$, and the corresponding Hamiltonian system is
\eqref{firstorder}.  In particular, $H$ is constant on any trajectory.  

Observe that $\Ot^d=\Om^d$, so that $\Ot$ and $\Om$ induce the same volume   
form.  Since Proposition~\ref{mainprop} only involves the
measure induced by $\Om$, it applies to give a description   
of the measure induced on a level set of $H$ by $|\Ot^d|$ and $H$, which 
is invariant under the flow $\varphi_t$ of $X_H$.  

We remark that by making a change of the momentum variable, \eqref{eom} can
also be realized as Hamilton's equations with respect to the canonical
symplectic form.  But in this realization the ``potential energy'' term in
the Hamiltonian function depends on both position and momentum.  See
\cite{MR}.

The same reasoning holds in the case of $n$ noninteracting free particles.
Let their positions be $x_k$ and their observed positions
be $x^F_k=(y^F_k,z^F_k,w^F_k)$, $1\leq k\leq n$.  Let
$v^F_k=\dot{x}^F_k$ denote the observed velocity of the $k$-th particle, and 
$p_k=m_kv^F_k$ its observed momentum.  Set
$\mathbf{x}^F=(x^F_1,\ldots,x^F_n)$,  
$\mathbf{v}^F=(v^F_1,v^F_2,\ldots,v^F_n)$.  Again take   
$M=\{\mathbf{x}^F\in\R^N\}$, $N=nd$, with coordinate $q=\mathbf{x}^F$, and
set $p=(p_1,\ldots,p_n)$.  The metric is again given by \eqref{metric}.      
Since the particles do 
not interact, \eqref{eom} holds with $x^F$ replaced by $x^F_k$ for each  
$k$.  The corresponding vector field $X$ is again given by \eqref{X} but
now with $q$, $p\in \R^N$.
The equations of motion are equivalent to Hamilton's equations for  
symplectic form on $T^*\R^N$ given by
\[
\Ot=\Om + 2\om\sum_{k=1}^n m_kdy^F_k\wedge dz^F_k,  
\]
where $\Om$ is the canonical symplectic form, and Hamiltonian 
\begin{equation}\label{HV}
H(q,p)=\frac12 \sum_{k=1}^n\frac{1}{m_k} \|p_k\|^2 +V(q),\qquad
V(q):=-\frac12\om^2\sum_{k=1}^nm_k \left[(y^F_k)^2+(z^F_k)^2\right].  
\end{equation}
The level surfaces of $H$ are given as usual by 
$\cS_E=\{(q,p):\tfrac12 |p|_g^2 +V(q)=E\}$.  
Note that $\cS_E^r=\cS_E$ for $E\neq 0$, while for $E=0$ one has 
$\cS_E\setminus \cS_E^r=\{(q,p):p=0, q=(x_1^F,\ldots,x_n^F)\text{ where }  
x_1^{F,H}=\ldots= x_n^{F,H}=0\}$.  
Also note that $\cS_E^0=\cS_E$ if $E>0$.    

As in the previous example, define $\Psi$ by \eqref{psi}, \eqref{pbar} and 
$\Phibar=\Psi\circ\Phi$.  Just as in Proposition~\ref{exampleprop}, the
measure defined by \eqref{Phimeasure}, with $V(q)$ now given by 
\eqref{HV}, is the restriction to $\cS_E^0$ of a smooth  
measure on $\cS_E^r$ ($= \cS_E$ if $E\neq 0$) which is invariant under the 
flow $\varphi_t|_{\cS_E^r}$.  In case $E=0$, this measure extends to an 
invariant measure on $\cS_E$ by requiring 
$\cS_E\setminus \cS_E^r$ to have measure $0$.  Summarizing, we have  
\begin{proposition}\label{rotatingprop}
The Hamiltonian \eqref{HV} is conserved for the system of $n$
noninteracting free particles in $\R^d$ viewed by an observer rotating at
constant angular velocity $\omega$.  The measure 
\[
\big(2(E-V(q))\big)^{\frac{N}{2}-1}\Phibar^*(d\sigma_1(\pb) dv(q))  
\]
is the
restriction to $\cS_E^0$ of a measure on $\cS_E$ which is 
invariant under the flow $\varphi_t|_{\cS_E}$, where $\pb$ is given by
\eqref{pbar}, $V(q)$ is given by \eqref{HV}, and $\varphi_t$ is the flow of
$X$ on $T^*\R^{nd}$.      
\end{proposition}

\section{Acknowledgments}

We are grateful to Martin Hairer, Robert Ho\l yst, Domokos Sz\'asz and
Balint Toth for very helpful advice.

\bibliographystyle{plain}

\end{document}